\documentclass{article}

\PassOptionsToPackage{numbers, square}{natbib}
\usepackage[final]{neurips_2021}

\usepackage[utf8]{inputenc} % allow utf-8 input
\usepackage[T1]{fontenc}    % use 8-bit T1 fonts
\usepackage{hyperref}       % hyperlinks
\usepackage{url}            % simple URL typesetting
\usepackage{booktabs}       % professional-quality tables
\usepackage{amsfonts}       % blackboard math symbols
\usepackage{nicefrac}       % compact symbols for 1/2, etc.
\usepackage{microtype}      % microtypography
\usepackage{xcolor}         % colors

\usepackage{amsmath}
\usepackage[ruled, linesnumbered]{algorithm2e}

\newcommand{\aref}[1]{\hyperref[#1]{Appendix~\ref*{#1}}}

\newcommand{\ie}{\textit{i}.\textit{e}.}
\newcommand{\eg}{\textit{e}.\textit{g}.}
\usepackage{forloop}
\usepackage{graphicx}
\usepackage{amsthm}
\usepackage{amssymb}
\usepackage{multirow}

\title{SNIPS: Solving Noisy Inverse Problems Stochastically}

\author{%
  Bahjat Kawar, Gregory Vaksman, Michael Elad\\
  Computer Science Department, Technion, Haifa, Israel \\
  \texttt{\{bahjat.kawar, grishav, elad\}@cs.technion.ac.il} \\
}

\begin{document}

\maketitle

\begin{abstract}
  In this work we introduce a novel stochastic algorithm dubbed \emph{SNIPS}, which draws samples from the posterior distribution of any linear inverse problem, where the observation is assumed to be contaminated by additive white Gaussian noise. Our solution incorporates ideas from Langevin dynamics and Newton's method, and exploits a pre-trained minimum mean squared error (MMSE) Gaussian denoiser. The proposed approach relies on an intricate derivation of the posterior score function that includes a singular value decomposition (SVD) of the degradation operator, in order to obtain a tractable iterative algorithm for the desired sampling. Due to its stochasticity, the algorithm can produce multiple high perceptual quality samples for the same noisy observation. We demonstrate the abilities of the proposed paradigm for image deblurring, super-resolution, and compressive sensing. We show that the samples produced are sharp, detailed and consistent with the given measurements, and their diversity exposes the inherent uncertainty in the inverse problem being solved.
\end{abstract}

\section{Introduction}
Many problems in the field of image processing can be cast as noisy linear inverse problems. This family of tasks includes denoising, inpainting, deblurring, super resolution, compressive sensing, and many other image recovery problems.
A general linear inverse problem is posed as
\begin{equation}
\label{eqn:intro}
    \mathbf{y} = \mathbf{Hx} + \mathbf{z},
\end{equation}
where we aim to recover a signal $\mathbf{x}$ from its measurement $\mathbf{y}$, given through a linear degradation operator $\mathbf{H}$ and a contaminating noise, being additive, white and Gaussian, $\mathbf{z} \sim \mathcal{N}\left(0, \sigma_0^2 \mathbf{I}\right)$. In this work we assume that both $\mathbf{H}$ and $\sigma_0$ are known.

Over the years, many strategies, algorithms and underlying statistical models were developed for handling image restoration problems. A key ingredient in many of the classic attempts is the prior that aims to regularize the inversion process and lead to visually pleasing results. Among the various options explored, we mention sparsity-inspired techniques~\cite{elad2006image,yang2010image,dong2012nonlocally}, local Gaussian-mixture modeling~\cite{yu2011solving,zoran2011learning}, and methods relying on non-local self-similarity~\cite{buades2005non, danielyan2011bm3d, ram2013image, vaksman2016patch}. More recently, and with the emergence of deep learning techniques, a direct design of the recovery path from $\mathbf{y}$ to an estimate of $\mathbf{x}$ took the lead, yielding state-of-the-art results in various linear inverse problems, such as denoising~\cite{lefkimmiatis2017non, zhang2017beyond, zhang2018ffdnet, vaksman2020lidia}, deblurring~\cite{kupyn2019deblurgan, suin2020spatially}, super resolution~\cite{dong2016accelerating, haris2018deep,wang2018esrgan} and other tasks~\cite{mccann2017convolutional, lucas2018using, hyun2018deep, gupta2018cnn, ravishankar2019image, zhang2018ista}.

Despite the evident success of the above techniques, many image restoration algorithms still have a critical shortcoming: In cases of severe degradation, most recovery algorithms tend to produce washed out reconstructions that lack details. Indeed, most image restoration techniques 
seek a reconstruction that minimizes the mean squared error between the restored image, $\mathbf{\hat{x}}$, and the unknown original one, $\mathbf{x}$. When the degradation is acute and information is irreversibly lost, image reconstruction becomes a highly ill-posed problem, implying that many possible clean images could explain the given measurements. The MMSE solution averages all these candidate solutions, being the conditional mean of the posterior of $\mathbf{x}$ given $\mathbf{y}$, leading to an image with loss of fine details in the majority of practical cases.
A recent work reported in~\cite{blau2018perception} has shown that reconstruction algorithms necessarily suffer from a perception-distortion tradeoff, \ie, targeting a minimization of the error between $\mathbf{\hat{x}}$ and $\mathbf{x}$ (in any metric) is necessarily accompanied by a compromised perceptual quality.
As a consequence, as long as we stick to the tendency to design recovery algorithms that aim for minimum MSE (or other distances), only a limited perceptual improvement can be expected. 

When perceptual quality becomes our prime objective, the strategy for solving inverse problems must necessarily change. More specifically, the solution should 
concentrate on producing a sample (or many of them) from the posterior distribution $p\left(\mathbf{x} | \mathbf{y}\right)$ instead of its conditional mean. 
Recently, two such approaches have been suggested -- GAN-based and Langevin sampling.
Generative Adversarial Networks (GANs) have shown impressive results in generating realistically looking images (\eg, ~\cite{goodfellow2014gans, radford2016dcgan}). GANs can be utilized for solving inverse problems while producing high-quality images (see \eg~\cite{bahat2020sr, menon2020pulse, peng2020generating}).
These solvers aim to produce a diverse set of output images that are consistent with the measurements, while also being aligned with the distribution of clean examples.
A major disadvantage of GAN-based algorithms for inverse problems is their tendency (as practiced in~\cite{bahat2020sr, menon2020pulse, peng2020generating}) to assume noiseless measurements, a condition seldom met in practice. An exception to this  is the work reported in~\cite{ohayon2021high}, which adapts a conditional GAN to become a stochastic denoiser.

The second approach for sampling from the posterior, and the one we shall be focusing on in this paper, is based on Langevin dynamics.
This core iterative technique enables sampling from a given distribution by leveraging the availability of the score function -- the gradient of the log of the probability density function~\cite{roberts1996exponential, besag2001markov}.
The work reported in~\cite{song2019generative, simoncelli, song2020score} utilizes the annealed Langevin dynamics method, both for image synthesis and for solving \emph{noiseless} inverse problems.\footnote{The work reported in~\cite{ho2020denoising, saharia2021image, li2021srdiff} and~\cite{guo2019agem, laumont2021bayesian} is also relevant to this discussion, but somewhat different. We shall specifically address these papers' content and its relation to our work in \autoref{sec:background}.} Their synthesis algorithm relies on an MMSE Gaussian denoiser (given as a neural network) for approximating a gradually blurred score function.
In their treatment of inverse problems, the conditional score remains tractable and manageable due to the noiseless measurements assumption.

The question addressed in this paper is the following: How can the above line of Langevin-based work be generalized for handling linear inverse problems, as in \autoref{eqn:intro}, in which the measurements are noisy? A partial and limited answer to this question has already been given in~\cite{kawar2021stochastic} for the tasks of image denoising and inpainting. The present work generalizes these (\cite{song2019generative, simoncelli, song2020score, kawar2021stochastic}) results, and introduces a systematic way for sampling from the posterior distribution of any given noisy linear inverse problem.
As we carefully show, this extension is far from being trivial, due to two prime reasons: (i) The involvement of the degradation operator $\mathbf{H}$, which poses a difficulty for establishing a relationship between the reconstructed image and the noisy observation; and (ii) The intricate connection between the measurements' and the synthetic annealed Langevin noise. Our proposed remedy is a decorrelation of the measurements equation via a singular value decomposition (SVD) of the operator $\mathbf{H}$, which decouples the dependencies between the measurements, enabling each to be addressed by an adapted iterative process. In addition, we define the annealing noise to be built as portions of the measurement noise itself, in a manner that facilitates a constructive derivation of the conditional score function.

Following earlier work~\cite{song2019generative, simoncelli, song2020score, kawar2021stochastic}, our algorithm is initialized with a random noise image, gradually converging to the reconstructed result, while following the direction of the log-posterior gradient, estimated using an MMSE denoiser.
Via a careful construction of the gradual annealing noise sequence, from very high values to low ones, the entries in the derived score switch mode. Those referring to non-zero singular values start by being purely dependent on the measurements, and then transition to incorporate prior information based on the denoiser. As for entries referring to zero singular values, their corresponding entries undergo a pure synthesis process based on the prior-only score function. Note that the denoiser blends values in the evolving sample, thus intermixing the influence of the gradient entries.
Our derivations include an analytical expression for a position-dependent step size vector, drawing inspiration from Newton’s method in optimization. This stabilizes the algorithm and is shown to be essential for its success.

We refer hereafter to our algorithm as \emph{SNIPS} (Solution of Noisy Inverse Problems Stochastically).
Observe that as we target to sample from the posterior distribution $p\left(\mathbf{x} | \mathbf{y}\right)$, different runs of SNIPS on the same input necessarily yield different results, all of which valid solutions to the given inverse problem. This should
not come as a surprise, as ill-posedness implies that there are multiple viable solutions for the same data, as has already been suggested in the context of super resolution~\cite{menon2020pulse, bahat2020sr, peng2020generating}. We demonstrate SNIPS on image deblurring, single image super resolution, and compressive sensing, all of which contain non-negligible noise, and emphasize the high perceptual quality of the results, their diversity, and their relation to the MMSE estimate.

To summarize, this paper's contributions are threefold:
\begin{itemize}
    \item We present an intricate derivation of the blurred posterior score function for general noisy inverse problems, where both the measurement and the target image contain delicately inter-connected additive white Gaussian noise.
    \item We introduce a novel stochastic algorithm -- \emph{SNIPS} -- that can sample from the posterior distribution of these problems. The algorithm relies on the availability of an MMSE denoiser.
    \item We demonstrate impressive results of \emph{SNIPS} on image deblurring, single image super resolution, and compressive sensing, all of which are highly noisy and ill-posed.
\end{itemize}

Before diving into the details of this work, we should mention that using Gaussian denoisers iteratively for handling general linear inverse problems has been already proposed in the context of the Plug-and-Play-Prior (PnP) method~\cite{venkatakrishnan2013plug} and RED~\cite{romano2017red}, and their many followup papers (\eg,~\cite{zhang2017learning, meinhardt2017learning, arridge2019solving, sun2019online, buzzard2018plug, tirer2018image, rond2016poisson, bigdeli2017deep}). However, both PnP and RED are quite different from our work, as they do not target sampling from the posterior, but rather focus on MAP or MMSE estimation.

\newcounter{row}
\newcounter{col}
\begin{figure}
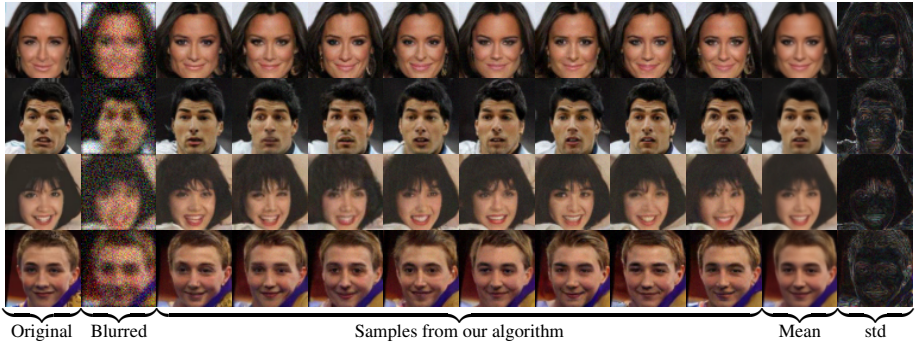

    \centering
    \def\arraystretch{0.1}
    \setlength\tabcolsep{0pt}
    \begin{tabular}{cccccccccccc}
         \forloop{row}{4}{\value{row} < 8}{
            \hspace{-0.05cm}\includegraphics[width=1cm]{./images/unideblur/sample_\arabic{row}_column_0.png} \hspace{-0.17cm}
            \forloop{col}{1}{\value{col} < 12}{
                & \includegraphics[width=1cm,height=1cm]{./images/unideblur/sample_\arabic{row}_column_\arabic{col}.png} \hspace{-0.17cm}
            } \\
        }
        & \multicolumn{11}{c}{\vspace{0.5mm}}\\
        \upbracefill & \upbracefill &
        \multicolumn{8}{c}{
            \upbracefill
        } &
        \upbracefill & \upbracefill
        \\
        & \multicolumn{11}{c}{\vspace{0.5mm}}\\
        \scriptsize{Original} & \scriptsize{Blurred} &
        \multicolumn{8}{c}{
            \scriptsize{Samples from our algorithm}
        } &
        \scriptsize{Mean} & \scriptsize{std}
    \end{tabular}
    \caption{Deblurring results on CelebA~\cite{liu2015celeba} images (uniform $5 \times 5$ blur and an additive noise with $\sigma_0=0.1$). Here and in all other shown figures, the standard deviation image is scaled by 4 for better visual inspection.}
    \label{fig:deblur_celeba}
\end{figure}

\section{Background}
\label{sec:background}
The Langevin dynamics algorithm~\cite{besag2001markov,roberts1996exponential} suggests sampling from a probability distribution $p\left(\mathbf{x}\right)$ using the iterative transition rule
\begin{equation}
\label{eqn:langevin_dynamics}
    \mathbf{x}_{t+1} = \mathbf{x}_t + \alpha \nabla_{\mathbf{x}_t} \log p\left(\mathbf{x}_t\right) + \sqrt{2\alpha} \mathbf{z}_t \;,
\end{equation}
where $\mathbf{z}_t \sim \mathcal{N}\left(0, \mathbf{I}\right)$ and $\alpha$ is an appropriately chosen small constant. The added $\mathbf{z}_t$ allows for stochastic sampling, avoiding a collapse to a maximum of the distribution. Initialized randomly, after a sufficiently large number of iterations, and under some mild conditions, this process converges to a sample from the desired distribution $p\left(\mathbf{x}\right)$~\cite{roberts1996exponential}.

The work reported in~\cite{song2019generative} extends the aforementioned algorithm into \textit{annealed Langevin dynamics}.
The annealing proposed replaces the score function in \autoref{eqn:langevin_dynamics} with a blurred version of it, $\nabla_{\mathbf{\tilde{x}}_t} \log p\left(\mathbf{\tilde{x}}_t\right)$, where ${\mathbf{\tilde{x}_t} = \mathbf{x}_t + \mathbf{n}}$ and $\mathbf{n} \sim \mathcal{N}\left(0, \sigma^{2} \mathbf{I}\right)$ is a synthetically injected noise. The core idea is to start with a very high noise level $\sigma$ and gradually drop it to near-zero, all while using a step size $\alpha$ dependent on the noise level. These changes allow the algorithm to converge much faster and perform better, because it widens the basin of attraction of the sampling process.
The work in~\cite{simoncelli} further develops this line of work by leveraging a brilliant relation attributed to Miyasawa~\cite{Miyasawa61} (also known as Stein's integration by parts trick~\cite{stein1981estimation} or Tweedie's identity~\cite{efron2011tweedie}). It is given as
\begin{equation}
\label{eqn:denoiser}
    \nabla_{\mathbf{\tilde{x}}_t} \log p\left(\mathbf{\tilde{x}}_t\right) = \frac{\mathbf{D}\left(\mathbf{\tilde{x}}_t, \sigma\right) - \mathbf{\tilde{x}}_t}{\sigma^2},
\end{equation}
where $\mathbf{D}\left(\mathbf{\tilde{x}}_t, \sigma\right) = \mathbb{E}\left[\mathbf{x} | \mathbf{\tilde{x}}_t\right]$ is the minimizer of the MSE measure $\mathbb{E}\left[\| \mathbf{x} - \mathbf{D}\left(\mathbf{\tilde{x}}_t, \sigma\right) \|_2^2\right]$, which can be approximated using a denoising neural network. This facilitates the use of denoisers in Langevin dynamics as a replacement for the evasive score function.

When turning to solve inverse problems, previous work suggests sampling from the posterior distribution $p\left(\mathbf{x} | \mathbf{y}\right)$ using annealed Langevin dynamics~\cite{simoncelli, song2020score, kawar2021stochastic} or similar methods~\cite{guo2019agem, ho2020denoising, saharia2021image, li2021srdiff}, by replacing the score function used in the generation algorithm with a conditional one.
As it turns out, if limiting assumptions can be posed on the measurements formation, the conditional score is tractable, and thus generalization of the annealed Langevin process to these problems is within reach. Indeed, in~\cite{song2019generative, simoncelli, song2020score, saharia2021image, li2021srdiff} the core assumption is $\mathbf{y}=\mathbf{Hx}$ for specific and simplified choices of $\mathbf{H}$ and with no noise in the measurements.
The works in~\cite{guo2019agem, laumont2021bayesian} avoid these difficulties altogether by returning to the original (non-annealed) Langevin method, with the unavoidable cost of becoming extremely slow. In addition, their algorithms are demonstrated on inverse problems in which the additive noise is restricted to be very weak.
The work in~\cite{kawar2021stochastic} is broader, allowing for an arbitrary additive white Gaussian noise, but limits $\mathbf{H}$ to the problems of denoising or inpainting.
While all these works demonstrate high quality results, there is currently no clear way for deriving the blurred score function of a general linear inverse problem as posed in \autoref{eqn:intro}. In the following, we present such a derivation.
\section{The Proposed Approach: Deriving the Conditional Score Function}
\subsection{Problem Setting}
\label{sec:problem}
We consider the problem of recovering a signal $\mathbf{x} \in \mathbb{R}^N$ (where $\mathbf{x} \sim p\left(\mathbf{x}\right)$ and $p\left(\mathbf{x}\right)$ is unknown) from the observation $\mathbf{y} = \mathbf{Hx} + \mathbf{z}$, where $\mathbf{y} \in \mathbb{R}^M, \mathbf{H} \in \mathbb{R}^{M \times N}, M \leq N,  \mathbf{z}\sim\mathcal{N}\left(0,\sigma_{0}^{2} \mathbf{I}\right)$, and $\mathbf{H}$ and $\sigma_0$ are known.\footnote{We assume $M \leq N$ for ease of notations, and because this is the common case. However, the proposed approach and all our derivations work just as well for $M > N$.}
Our ultimate goal is to sample from the posterior $p\left(\mathbf{x}|\mathbf{y}\right)$. However, since access to the score function $\nabla_\mathbf{x} \log p(\mathbf{x}|\mathbf{y})$ is not available, we retarget our goal, as explained above, to sampling from blurred posterior distributions, $p\left(\mathbf{\tilde{x}}|\mathbf{y}\right)$, where ${\mathbf{\tilde{x}} = \mathbf{x} + \mathbf{n}}$ and $\mathbf{n} \sim \mathcal{N}\left(0, \sigma^{2}\mathbf{I}\right)$, with noise levels $\sigma$ starting very high, and decreasing towards near-zero. 

As explained in \aref{sec:plus_appendix}, the sampling should be performed in the SVD domain in order to get a tractable derivation of the blurred score function. Thus,
we consider the singular value decomposition (SVD) of $\mathbf{H}$, given as $\mathbf{H}=\mathbf{U} \mathbf{\Sigma} \mathbf{V}^T$, where $\mathbf{U} \in \mathbb{R}^{M \times M}$ and $\mathbf{V} \in \mathbb{R}^{N \times N}$ are orthogonal matrices, and $\mathbf{\Sigma} \in \mathbb{R}^{M \times N}$ is a rectangular diagonal matrix containing the singular values of $\mathbf{H}$, denoted as $\left\{s_j\right\}_{j=1}^{M}$ in descending order (${s_1 > s_2 > \dots > s_{M-1} > s_M \ge 0}$). For convenience of notations, we also define $s_{j} = 0$ for $j = M+1, \dots, N$.
To that end, we notice that
\begin{equation}
    p\left(\mathbf{\tilde{x}} | \mathbf{y}\right) = p\left(\mathbf{\tilde{x}} | \mathbf{U}^T \mathbf{y}\right) = p\left(\mathbf{V}^T \mathbf{\tilde{x}} | \mathbf{U}^T \mathbf{y}\right).
\end{equation}
The first equality holds because the multiplication of $\mathbf{y}$ by the orthogonal matrix $\mathbf{U}^T$ does not add or remove information, and the second equality holds because the multiplication of $\mathbf{\tilde x}$ by $\mathbf{V}^T$ does not change its probability~\cite{prob}. Therefore, sampling from $p\left(\mathbf{V}^{T} \mathbf{\tilde{x}} | \mathbf{U}^T \mathbf{y}\right)$ and then multiplying the result by $\mathbf{V}$ will produce the desired sample from $p\left(\mathbf{\tilde{x}} | \mathbf{y}\right)$. As we are using Langevin dynamics, we need to calculate the conditional score function $\nabla_{\mathbf{V}^{T} \mathbf{\tilde{x}}} \log p\left(\mathbf{V}^{T} \mathbf{\tilde{x}} | \mathbf{U}^T \mathbf{y}\right)$.
For simplicity, we denote hereafter $\mathbf{y}_T = \mathbf{U}^T \mathbf{y}$, $\mathbf{z}_T = \mathbf{U}^T \mathbf{z}$, $\mathbf{x}_T = \mathbf{V}^T \mathbf{x}$, $\mathbf{n}_T = \mathbf{\Sigma V}^T \mathbf{n}$, and $\mathbf{\tilde{x}}_T = \mathbf{V}^T \mathbf{\tilde{x}}$.
Observe that with these notations, the measurements equation becomes
\[
\mathbf{y} = \mathbf{Hx}+\mathbf{z} = \mathbf{U \Sigma} \mathbf{V}^T \mathbf{x}+\mathbf{z},
\]
and thus
\[
    \mathbf{U}^T \mathbf{y} = \mathbf{\Sigma} \mathbf{V}^T \mathbf{x} +\mathbf{U}^T\mathbf{z}
          = \mathbf{\Sigma} \mathbf{V}^T (\mathbf{\tilde{x}} - \mathbf{n}) +\mathbf{U}^T \mathbf{z}
          = \mathbf{\Sigma} \mathbf{V}^T \mathbf{\tilde x} - \mathbf{\Sigma} \mathbf{V}^T \mathbf{n} + \mathbf{U}^T \mathbf{z}, 
\]
where we have relied on the relation $\mathbf{\tilde x} = \mathbf{x} + \mathbf{n}$. This leads to 
\begin{equation}
\label{eqn:measurements}
    \mathbf{y}_T = \mathbf{\Sigma} \mathbf{\tilde x}_T - \mathbf{n}_T + \mathbf{z}_T.
\end{equation}

In this formulation, which will aid in deriving the conditional score, our aim is to make design choices on $\mathbf{n}_T$ such that $\mathbf{z}_T-\mathbf{n}_T$ has uncorrelated entries and is independent of $\mathbf{\tilde x}_T$. This brings us to the formation of the synthetic annealed noise, which is an intricate ingredient in our derivations. 

We base this formation on the definition of a sequence of noise levels $\left\{\sigma_i\right\}_{i=1}^{L+1}$ such that ${\sigma_{1} > \sigma_{2} > \dots > \sigma_{L} > \sigma_{L+1}=0}$, where $\sigma_{1}$ is high (possibly $\sigma_1 > \left\lVert\mathbf{x}\right\rVert_\infty$) and $\sigma_{L}$ is close to zero.
We require that for every $j$ such that $s_j \neq 0$, there exists $i_j$ such that ${\sigma_{i_j} s_j < \sigma_0}$ and ${\sigma_{i_j - 1} s_j > \sigma_0}$. This implies $\forall i: \sigma_i s_j \neq \sigma_0$, which helps ease notations. SNIPS works just as well for $\sigma_i s_j = \sigma_0$.

Using $\left\{\sigma_i\right\}_{i=1}^{L+1}$, we would like to define $\left\{\mathbf{\tilde{x}}_i\right\}_{i=1}^{L+1}$, a sequence of noisy versions of $\mathbf{x}$, where the noise level in $\mathbf{\tilde{x}}_i$ is $\sigma_i$.
One might be tempted to define these noise additions as independent of the measurement noise $\mathbf{z}$. However, this option leads to a conditional score term that cannot be calculated analytically, as explained in~\aref{sec:plus_appendix}.
Therefore, we define these noise additions differently, as carved from $\mathbf{z}$ in a gradual fashion.
To that end, we define $\mathbf{\tilde{x}}_{L+1} = \mathbf{x}$, and for every ${i = L, L-1, \dots, 1}$: ${\mathbf{\tilde{x}}_{i} = \mathbf{\tilde{x}}_{i+1} + \text{\boldmath$\eta$}_i}$, where ${\text{\boldmath$\eta$}_i \sim \mathcal{N}\left(0, \left( \sigma_i^2 - \sigma_{i+1}^2 \right) \mathbf{I}\right)}$.
This results in $\mathbf{\tilde{x}}_{i} = \mathbf{x} + \mathbf{n}_i$, where ${\mathbf{n}_i = \sum_{k=i}^{L} \text{\boldmath$\eta$}_k \sim \mathcal{N}\left(0, \sigma_i^2 \mathbf{I}\right)}$.

And now we turn to define the statistical dependencies between the measurements' noise $\mathbf{z}$ and the artificial noise vectors $\text{\boldmath$\eta$}_i$. Since $\text{\boldmath$\eta$}_i$ and $\mathbf{z}$ are each Gaussian with uncorrelated entries, so are the components of the vectors $\mathbf{\Sigma} \mathbf{V}^T \text{\boldmath$\eta$}_i$, $\mathbf{\Sigma} \mathbf{V}^T \mathbf{n}_i$, and $\mathbf{z}_T$. In order to proceed while easing notations, let us focus on a single entry $j$ in these three vectors, for which $s_j>0$, and omit this index. We denote these entries as ${\eta}_{T,i}$, $n_{T, i}$ and ${z}_T$, respectively. We construct $\eta_{T,i}$ such that   
\[
\mathbb{E}\left[\eta_{T,i} \cdot z_T\right] = \begin{cases}
\mathbb{E}\left[\eta_{T,i}^2\right] & \text{for } i \geq i_j \\
\mathbb{E}\left[\left(z_T - n_{T,i_j}\right)^2\right] & \text{for } i = i_j - 1 \\
0 & \text{otherwise.}
\end{cases}
\]
This implies that the layers of noise $\eta_{T,L+1}, \dots \eta_{T,i_j}$ are all portions of $z_T$ itself, with an additional portion being contained in $\eta_{T, i_j - 1}$. Afterwards, $\eta_{T, i}$ become independent of $z_T$.
In the case of $s_j=0$, the above relations simplify to be ${E[\eta_{T,i} \cdot z_T] = 0}$ for all $i$, implying no statistical dependency between the given and the synthetic noises.
Consequently, it can be shown that the overall noise in \autoref{eqn:measurements} satisfies
\begin{equation}
\label{eqn:relation_to_z}
    \left( \mathbf{\Sigma} \mathbf{V}^T \mathbf{n}_i - \mathbf{z_T} \right)_j = n_{T,i} - z_T \sim
    \begin{cases}
        \mathcal{N}\left(0, s_j^2 \sigma_{i}^2 - \sigma_0^2\right) & \text{if $\sigma_i s_j > \sigma_0$} \\
        \mathcal{N}\left(0, \sigma_0^2 - s_j^2 \sigma_{i}^2\right) & \text{otherwise.} \\
    \end{cases}
\end{equation}
The top option refers to high values of the annealed Langevin noise, in which, despite the possible decay caused by the singular value $s_j$, this noise is stronger than $z_T$. In this case, $n_{T,i}$ contains all $z_T$ and an additional independent portion of noise. The bottom part assumes that the annealed noise (with the influence of $s_j$) is weaker than the measurements' noise, and then it is fully immersed within $z_T$, with the difference being Gaussian and independent. 

\subsection{Derivation of the Conditional Score Function}
\label{sec:grad_calc}
The above derivations show that the noise in \autoref{eqn:measurements} is zero-mean, Gaussian with uncorrelated entries and of known variance, and this noise is independent of $\mathbf{\tilde{x}}_i$. Thus \autoref{eqn:measurements} can be used conveniently for deriving the measurements part of the conditional score function. 
We denote $\mathbf{\tilde{x}}_T = \mathbf{V}^T \mathbf{\tilde{x}}_i$, $\mathbf{\tilde{x}} = \mathbf{\tilde{x}}_i$,  $\mathbf{n} = \mathbf{n}_i$ for simplicity, and turn to calculate $\nabla_{ \mathbf{\tilde{x}}_T} \log p\left(\mathbf{\tilde{x}}_T | \mathbf{y}_T\right)$.
We split $\mathbf{\tilde{x}}_T$ into three parts:
(i) $\mathbf{\tilde{x}}_{T,0}$ refers to the entries $j$ for which $s_j = 0$;
(ii) $\mathbf{\tilde{x}}_{T,<}$ corresponds to the entries $j$ for which $0 < \sigma_{i} s_j < \sigma_0$;
and (iii) $\mathbf{\tilde{x}}_{T,>}$ includes the entries $j$ for which $ \sigma_{i} s_j > \sigma_0$.
Observe that this partition of the entries of $\mathbf{\tilde{x}}_T$ is non-overlapping and fully covering.
Similarly, we partition every vector $\mathbf{v} \in \mathbb{R}^N$ into $\mathbf{v}_0, \mathbf{v}_<, \mathbf{v}_>$, which are the entries of $\mathbf{v}$ corresponding to $\mathbf{\tilde{x}}_{T,0}, \mathbf{\tilde{x}}_{T,<}, \mathbf{\tilde{x}}_{T,>}$, respectively.
Furthermore, we define $\mathbf{v}_{\not{0}}, \mathbf{v}_{\not{<}}, \mathbf{v}_{\not{>}}$ as all the entries of $\mathbf{v}$ except $\mathbf{v}_0, \mathbf{v}_<, \mathbf{v}_>$, respectively.
With these definitions in place, the complete derivation of the score function is detailed in \aref{sec:proofs}, and here we bring the final outcome. For $\mathbf{\tilde{x}}_{T,0}$, the score is independent of the measurements and given by
\begin{equation}
\label{eqn:grad_0}
    \nabla_{ \mathbf{\tilde{x}}_{T,0}} \log p\left(\mathbf{\tilde{x}}_T | \mathbf{y}_T\right) = \left(\mathbf{V}^T \nabla_{\mathbf{\tilde{x}}} \log p\left(\mathbf{\tilde{x}}\right) \right)_{0}.
\end{equation}

For the case of $\mathbf{\tilde{x}}_{T,>}$, the expression obtained is only measurements-dependent,
\begin{equation}
\label{eqn:grad_b}
     \nabla_{ \mathbf{\tilde{x}}_{T,>} } \log p\left(\mathbf{\tilde{x}}_T | \mathbf{y}_T\right) =
    \left( \mathbf{\Sigma}^T \left( \sigma_i^2 \mathbf{\Sigma \Sigma}^T - \sigma_0^2 \mathbf{I} \right)^{\dagger} \left( \mathbf{y}_T - \mathbf{\Sigma} \mathbf{\tilde{x}}_T \right) \right)_>.
\end{equation}

Lastly, for the case of $\mathbf{\tilde{x}}_{T,<}$, the conditional score includes two terms -- one referring to the plain (blurred) score, and the other depending on the measurements, 
\begin{equation}
\label{eqn:grad_s}
    \nabla_{ \mathbf{\tilde{x}}_{T,<}} \log p\left(\mathbf{\tilde{x}}_T | \mathbf{y}_T\right) =
    \left( \mathbf{\Sigma}^T \left( \sigma_0^2 \mathbf{I} - \sigma_i^2 \mathbf{\Sigma \Sigma}^T \right)^{\dagger} \left( \mathbf{y}_T - \mathbf{\Sigma}  \mathbf{\tilde{x}}_T \right) \right)_< + \left(\mathbf{V}^T \nabla_{\mathbf{\tilde{x}}} \log p\left(\mathbf{\tilde{x}}\right) \right)_<.
\end{equation}

As already mentioned, the full derivations of equations \ref{eqn:grad_0}, \ref{eqn:grad_b}, and \ref{eqn:grad_s} are detailed in \autoref{sec:proofs}. Aggregating all these results together, we obtain the following conditional score function:
\begin{equation}
\label{eqn:gradient}
    \nabla_{ \mathbf{\tilde{x}}_T } \log p\left(\mathbf{\tilde{x}}_T | \mathbf{y}_T\right) =
    \mathbf{\Sigma}^T \left| \sigma_0^2 \mathbf{I} - \sigma_i^2 \mathbf{\Sigma \Sigma}^T \right|^{\dagger}
    \left( \mathbf{y}_T - \mathbf{\Sigma}  \mathbf{\tilde{x}}_T \right) +
    \left. \left( \mathbf{V}^T \nabla_{\mathbf{\tilde{x}}} \log p\left(\mathbf{\tilde{x}}\right) \right) \right|_{\not{>}},
\end{equation}
where $\left. \left( \mathbf{v} \right) \right|_{\not{>}}$ is the vector $\mathbf{v}$, but with zeros in its entries that correspond to $\mathbf{v}_{>}$. Observe that the first term in \autoref{eqn:gradient} contains zeros in the entries corresponding to $\mathbf{\tilde{x}}_{T, 0}$, matching the above calculations.
The vector $\nabla_{\mathbf{\tilde{x}}} \log p\left(\mathbf{\tilde{x}}\right)$ can be estimated using a neural network as in~\cite{song2019generative}, or using a pre-trained MMSE denoiser as  in~\cite{simoncelli, kawar2021stochastic}.
All the other elements of this vector are given or can be easily obtained from $\mathbf{H}$ by calculating its SVD decomposition once at the beginning.

\section{The Proposed Algorithm}
\begin{figure}
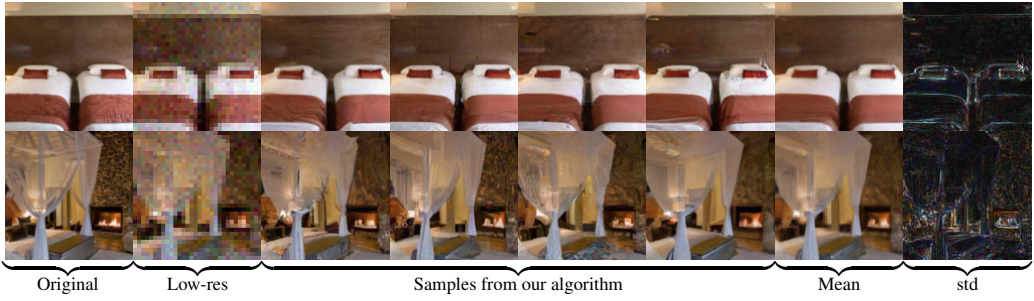

    \centering
    \def\arraystretch{0.1}
    \setlength\tabcolsep{0pt}
    \begin{tabular}{cccccccc}
        \forloop{row}{0}{\value{row} < 2}{
            \hspace{-0.05cm}\includegraphics[width=1.7cm,height=1.7cm]{./images/bedroom_sr4/sample_\arabic{row}_column_0.png} \hspace{-0.17cm}
            \forloop{col}{1}{\value{col} < 8}{
                & \includegraphics[width=1.7cm,height=1.7cm]{./images/bedroom_sr4/sample_\arabic{row}_column_\arabic{col}.png} \hspace{-0.17cm}
            } \\
        }
        & \multicolumn{7}{c}{\vspace{0.5mm}}\\
        \upbracefill & \upbracefill &
        \multicolumn{4}{c}{
            \upbracefill
        } &
        \upbracefill & \upbracefill
        \\
        & \multicolumn{7}{c}{\vspace{0.5mm}}\\
        \scriptsize{Original} & \scriptsize{Low-res} &
        \multicolumn{4}{c}{
            \scriptsize{Samples from our algorithm}
        } &
        \scriptsize{Mean} & \scriptsize{std}
    \end{tabular}
    \caption{Super resolution results on LSUN bedroom~\cite{yu2015lsun} images (downscaling $4:1$ by plain averaging and adding noise with $\sigma_0 = 0.04$).}
    \label{fig:sr4_bedroom}
\end{figure}
Armed with the conditional score function in \autoref{eqn:gradient}, the Langevin dynamics algorithm can be run with a constant step size or an annealed step size as in~\cite{song2019generative}, and this should converge to a sample from $p\left(\mathbf{\tilde{x}}_T | \mathbf{y}_T\right)$. However, for this to perform well, one should use a very small step size, implying a devastatingly slow convergence behavior.
This is mainly due to the fact that different entries of $\mathbf{\tilde{x}}_T$ advance at different speeds, in accord with their corresponding singular values. As the added noise in each step has the same variance in every entry, this leads to an unbalanced signal-to-noise ratio, which considerably slows down the algorithm.

In order to mitigate this problem, we suggest using a \textit{step size vector} $\text{\boldmath$\alpha$}_i \in \mathbb{R}^N$. We denote ${\mathbf{A}_i = {diag}\left(\text{\boldmath$\alpha$}_i\right)}$, and obtain the following update formula for a Langevin dynamics algorithm:
\begin{equation}
\label{eqn:update_formula}
    \mathbf{V}^T \mathbf{\tilde x}_i = \mathbf{V}^T \mathbf{\tilde x}_{i-1} + c \cdot \mathbf{A}_i \cdot \nabla_{\mathbf{V}^T \mathbf{\tilde x}_i} \log p\left(\mathbf{V}^T \mathbf{\tilde x}_i | \mathbf{y}_T\right) + \sqrt{2\cdot c}  \mathbf{A}_i^\frac{1}{2} \cdot \mathbf{z}_i,
\end{equation}
where the conditional score function is estimated as described in \autoref{sec:grad_calc}, and $c$ is some constant.
For the choice of the step sizes in the diagonal of $\mathbf{A}_i$, we draw inspiration from Newton's method in optimization, which is designed to speed up convergence to local maximum points. The update formula in Newton's method is the same as \autoref{eqn:update_formula}, but without the additional noise $\mathbf{z}_i$, and with $\mathbf{A}_i$ being the negative inverse Hessian of $\log p\left(\mathbf{V}^T \mathbf{\tilde x}_i | \mathbf{y}_T\right)$.
We calculate a diagonal approximation of the Hessian, and set $\mathbf{A}_i$ to be its negative inverse. We also estimate the conditional score function using \autoref{eqn:gradient} and a neural network.
Note that this mixture of Langevin dynamics and Newton's method has been suggested in a slightly different context in~\cite{simsekli2016stochastic}, where the Hessian was approximated using a Quasi-Newton method. In our case, we analytically calculate a diagonal approximation of the negative inverse Hessian and obtain the following:
\begin{equation}
    \left(\text{\boldmath$\alpha$}_i\right)_j = \begin{cases}
    \sigma_i^2, & s_j = 0 \\
    \sigma_i^2 - \frac{\sigma_0^2}{s_j^2}, & \sigma_{i} s_j > \sigma_0 \\
    \sigma_i^2 \cdot \left( 1 - s_j^2 \frac{\sigma_i^2}{\sigma_0^2} \right), & 0 < \sigma_{i} s_j < \sigma_0.
    \end{cases}
\end{equation}
The full derivations for each of the three cases are detailed in \aref{sec:proof_step}. Using these step sizes, the update formula in \autoref{eqn:update_formula}, the conditional score function in \autoref{eqn:gradient}, and a neural network $\mathbf{s}\left(\mathbf{\tilde{x}}, \sigma\right)$ that estimates the score function $\nabla_{\mathbf{\tilde{x}}} \log p\left( \mathbf{\tilde{x}} \right)$,\footnote{Recall that $\mathbf{s}\left(\mathbf{\tilde{x}}, \sigma\right)=\left(\mathbf{D}\left(\mathbf{\tilde{x}}, \sigma\right)-\mathbf{\tilde{x}}\right)/\sigma^2$, being a denoising residual.} we obtain a tractable iterative algorithm for sampling from $p\left( \mathbf{\tilde{x}}_L \mid \mathbf{y} \right)$, where the noise in $\mathbf{\tilde{x}}_L$ is sufficiently negligible to be considered as a sampling from the ideal image manifold.

\begin{algorithm}[H]
\label{alg:general}
\caption{SNIPS}
 \KwIn{$\left\{\sigma_{i}\right\}_{i=1}^{L}$, $c$, $\tau$, $\mathbf{y}$, $\mathbf{H}$, $\sigma_0$}
 $\mathbf{U}, \mathbf{\Sigma}, \mathbf{V} \leftarrow svd\left(\mathbf{H}\right)$ \\
 Initialize $\mathbf{x_0}$ with random noise $U\left[0, 1\right]$\\
 \For{$i$ $\leftarrow$ $1$ to $L$}{
   $\left( \mathbf{A}_{i} \right)_0 \leftarrow \sigma_{i}^{2} \mathbf{I}$ \\
   $\left( \mathbf{A}_{i} \right)_< \leftarrow \sigma_i^2 \cdot \left( \mathbf{I} - \frac{\sigma_i^2}{\sigma_0^2} \mathbf{\Sigma_< \Sigma_<}^T \right)$ \\
   $\left( \mathbf{A}_{i} \right)_> \leftarrow \sigma_i^2 \mathbf{I} - \sigma_0^2 \mathbf{\Sigma}_>^{\dagger} \mathbf{\Sigma}_>^{\dagger^T}$ \\
   \For{$t$ $\leftarrow$ $1$ to $\tau$}{
     Draw $\mathbf{z}_{t} \sim \mathcal{N}\left(0, \mathbf{I}\right)$ \\
     $\mathbf{d}_t \leftarrow 
     \mathbf{\Sigma}^T \cdot \left| \sigma_0^2 \mathbf{I} - \sigma_i^2 \mathbf{\Sigma \Sigma}^T \right|^{\dagger} \cdot
     \left( \mathbf{U}^T \mathbf{y} - \mathbf{\Sigma V}^T \mathbf{x}_{t-1} \right) +
     \left. \left( \mathbf{V}^T \cdot \mathbf{s}\left(\mathbf{x}_{t-1}, \sigma_{i}\right) \right) \right|_{\not{>}}$ \\
     $\mathbf{x}_t \leftarrow \mathbf{V} \cdot \left( \mathbf{V}^T \mathbf{x}_{t-1} + c \mathbf{A}_{i} \mathbf{d}_{t} + \sqrt{2c} \mathbf{A}_i^\frac{1}{2} \mathbf{z}_t \right)$
   }
 $\mathbf{x}_0 \leftarrow \mathbf{x}_\tau$
 }
 \KwOut{$\mathbf{x_0}$}
\end{algorithm}
Note that when we set $\mathbf{H} = 0$ and $\sigma_0=0$, implying no measurements, the above algorithm degenerates to an image synthesis, exactly as in~\cite{song2019generative}. Two other special cases of this algorithm are obtained for $\mathbf{H}=\mathbf{I}$ or $\mathbf{H}=\mathbf{I}$ with some rows removed, the first referring to denoising and the second to noisy inpainting, both cases shown in~\cite{kawar2021stochastic}. Lastly, for the choices of $\mathbf{H}$ as in~\cite{simoncelli} or~\cite{song2019generative, song2020score} and with $\sigma_0 =0$, the above algorithm collapses to a close variant of their proposed iterative methods.
\section{Experimental Results}
\label{sec:experiments}
\begin{figure}
    \centering
    \newcommand{\sizeee}{1.2cm}
    \def\arraystretch{0.2}
    \setlength\tabcolsep{0.2pt}
    \begin{tabular}{cccccc}
        &
        \includegraphics[width=\sizeee,height=\sizeee]{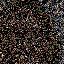} &
        \includegraphics[width=\sizeee,height=\sizeee]{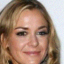} &
        \includegraphics[width=\sizeee,height=\sizeee]{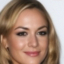} &
        \includegraphics[width=\sizeee,height=\sizeee]{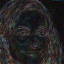} &
        \hspace{0.1cm}\raisebox{0.7cm}[0pt][0pt]{\rotatebox[origin=c]{90}{\scriptsize{$25\%$}}} \\
        
        \includegraphics[width=\sizeee,height=\sizeee]{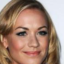} &
        \includegraphics[width=\sizeee,height=\sizeee]{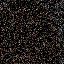} &
        \includegraphics[width=\sizeee,height=\sizeee]{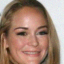} &
        \includegraphics[width=\sizeee,height=\sizeee]{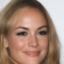} &
        \includegraphics[width=\sizeee,height=\sizeee]{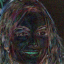} &
        \hspace{0.1cm}\raisebox{0.7cm}[0pt][0pt]{\rotatebox[origin=c]{90}{\scriptsize{$12.5\%$}}} \\

        &
        \includegraphics[width=\sizeee,height=\sizeee]{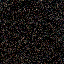} &
        \includegraphics[width=\sizeee,height=\sizeee]{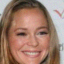} &
        \includegraphics[width=\sizeee,height=\sizeee]{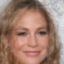} &
        \includegraphics[width=\sizeee,height=\sizeee]{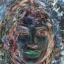} &
        \hspace{0.1cm}\raisebox{0.7cm}[0pt][0pt]{\rotatebox[origin=c]{90}{\scriptsize{$6.25\%$}}} \\
        
        \scriptsize{Original} & \scriptsize{Degraded} & \scriptsize{Sample} & \scriptsize{Mean} & \scriptsize{std} &
    \end{tabular}
    \caption{Compressive sensing results on a CelebA~\cite{liu2015celeba} image with an additive noise of $\sigma_0 = 0.1$.}
    \label{fig:cs_celeba}
\end{figure}
In our experiments we use the NCSNv2~\cite{song2020improved} network in order to estimate the score function of the prior distribution. Three different NCSNv2 models are used, each trained separately on training sets of: (i) images of size $64 \times 64$ pixels from the CelebA dataset~\cite{liu2015celeba}; (ii) images of size $128 \times 128$ pixels from LSUN~\cite{yu2015lsun} bedrooms dataset; and (iii) LSUN $128 \times 128$ images of towers.
We demonstrate SNIPS' capabilities on the respective test sets for image deblurring, super resolution, and compressive sensing. In each of the experiments, we run our algorithm 8 times, producing 8 samples for each input. We examine both the samples themselves and their mean, which serves as an approximation of the MMSE solution, $\mathbb{E}\left[\mathbf{x} | \mathbf{y}\right]$.

\textbf{For image deblurring}, we use a uniform $5 \times 5$ blur kernel, 
and an additive white Gaussian noise with $\sigma_0=0.1$ (referring to pixel values in the range $[0,1]$). \autoref{fig:deblur_celeba} demonstrates the obtained results for several images taken from the CelebA dataset. As can be seen, SNIPS produces visually pleasing, diverse samples.

\textbf{For super resolution}, the images are downscaled using a block averaging filter, \ie, each non-overlapping block of pixels in the original image is averaged into one pixel in the low-resolution image. We use blocks of size $2 \times 2$ or $4 \times 4$ pixels, and assume the low-resolution image to include an additive white Gaussian noise. We showcase results on LSUN and CelebA in Figures \ref{fig:sr4_bedroom}, \ref{fig:sr4_celeba}, and \ref{fig:sr2_celeba}.

\begin{figure}
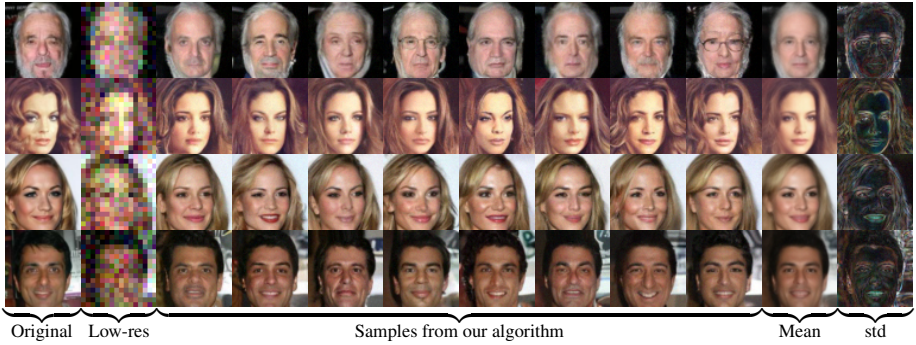

    \centering
    \def\arraystretch{0.1}
    \setlength\tabcolsep{0pt}
    \begin{tabular}{cccccccccccc}
        \forloop{row}{0}{\value{row} < 4}{
            \hspace{-0.05cm}\includegraphics[width=1cm,height=1cm]{./images/sr4/sample_\arabic{row}_column_0.png} \hspace{-0.17cm}
            \forloop{col}{1}{\value{col} < 12}{
                & \includegraphics[width=1cm,height=1cm]{./images/sr4/sample_\arabic{row}_column_\arabic{col}.png} \hspace{-0.17cm}
            } \\
        }
        & \multicolumn{11}{c}{\vspace{0.5mm}}\\
        \upbracefill & \upbracefill &
        \multicolumn{8}{c}{
            \upbracefill
        } &
        \upbracefill & \upbracefill
        \\
        & \multicolumn{11}{c}{\vspace{0.5mm}}\\
        \scriptsize{Original} & \scriptsize{Low-res} &
        \multicolumn{8}{c}{
            \scriptsize{Samples from our algorithm}
        } &
        \scriptsize{Mean} & \scriptsize{std}
    \end{tabular}
    \caption{Super resolution results on CelebA~\cite{liu2015celeba} images (downscaling $4:1$ by plain averaging and adding noise with $\sigma_0 = 0.1$).}
    \label{fig:sr4_celeba}
\end{figure}
\begin{figure}
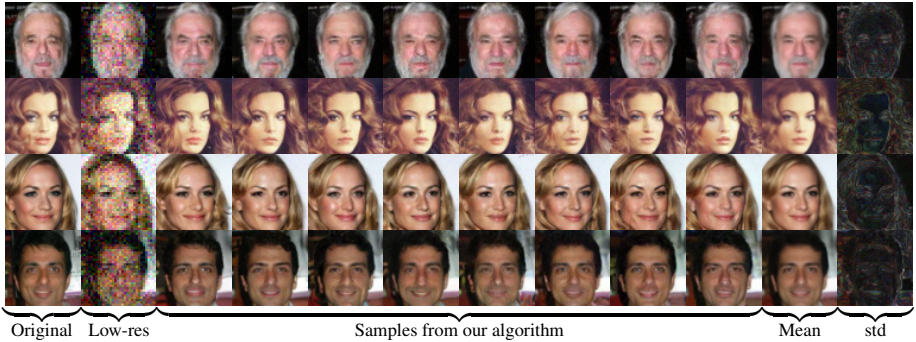

    \centering
    \def\arraystretch{0.1}
    \setlength\tabcolsep{0pt}
    \begin{tabular}{cccccccccccc}
        \forloop{row}{0}{\value{row} < 4}{
            \hspace{-0.05cm}\includegraphics[width=1cm,height=1cm]{./images/sr2/sample_\arabic{row}_column_0.png} \hspace{-0.17cm}
            \forloop{col}{1}{\value{col} < 12}{
                & \includegraphics[width=1cm,height=1cm]{./images/sr2/sample_\arabic{row}_column_\arabic{col}.png} \hspace{-0.17cm}
            } \\
        }
        & \multicolumn{11}{c}{\vspace{0.5mm}}\\
        \upbracefill & \upbracefill &
        \multicolumn{8}{c}{
            \upbracefill
        } &
        \upbracefill & \upbracefill
        \\
        & \multicolumn{11}{c}{\vspace{0.5mm}}\\
        \scriptsize{Original} & \scriptsize{Low-res} &
        \multicolumn{8}{c}{
            \scriptsize{Samples from our algorithm}
        } &
        \scriptsize{Mean} & \scriptsize{std}
    \end{tabular}
    \caption{Super resolution results on CelebA~\cite{liu2015celeba} images (downscaling $2:1$ by plain averaging and adding noise with $\sigma_0 = 0.1$).}
    \label{fig:sr2_celeba}
\end{figure}

\textbf{For compressive sensing}, we use three random projection matrices with singular values of $1$, that compress the image by $25\%$, $12.5\%$, and $6.25\%$. As can be seen in \autoref{fig:cs_celeba} and as expected, the more aggressive the compression, the more significant are the variations in reconstruction.

We calculate the average PSNR (peak signal-to-noise ratio) of each of the $8$ samples in our experiments, as well as the PSNR of their mean, as shown in \autoref{tab:psnrs}. In all the experiments, the empirical conditional mean presents an improvement of around $2.4$ dB in PSNR, even though it is less visually appealing compared to the samples.
This is consistent with the theory  in~\cite{blau2018perception}, which states that the difference in PSNR between posterior samples and the conditional mean (the MMSE estimator) should be $3$dB, with the MMSE estimator having poorer perceptual quality but better PSNR.

A comparison of our deblurring results to those obtained by RED~\cite{romano2017red} is detailed in \aref{sec:red}. We show that SNIPS exhibits superior performance over RED, achieving more than $11\%$ improvement in PSNR and more than $58\%$ improvement in LPIPS~\cite{lpips}, a perceptual quality metric.

\begin{table}
    \centering
    \caption{PSNR results for different inverse problems on 8 images from CelebA~\cite{liu2015celeba}. We ran SNIPS 8 times, and obtained 8 samples. The average PSNR for each of the samples is in the first column, while the average PSNR for the mean of the 8 samples for each image is in the second one.}
    \label{tab:psnrs}
    \begin{tabular}{ c  c c }
        \toprule
        \textbf{Problem} & \textbf{Sample PSNR} & \textbf{Mean PSNR} \\
        \midrule
        %Gaussian deblurring (of width $1$) & $26.95$ & $29.53$ \\
        %Gaussian deblurring (of width $10$) & $25.59$ & $28.01$ \\
        Uniform deblurring & $25.54$ & $28.01$ \\
        Super resolution (by $2$) & $25.58$ & $28.03$ \\
        Super resolution (by $4$) & $21.90$ & $24.31$ \\
        Compressive sensing (by $25\%$) & $25.68$ & $28.06$ \\
        Compressive sensing (by $12.5\%$) & $22.34$ & $24.67$ \\
        \bottomrule
    \end{tabular}
\end{table}

\subsection{Assessing Faithfulness to the Measurements}
A valid solution to an inverse problem should satisfy two conditions: (i) It should be visually pleasing, consistent with the underlying prior distribution of images, and (ii) It should be faithful to the given measurement, maintaining the relationship as given in the problem setting.
Since the prior distribution is unknown, we assess the first condition by visually observing the obtained solutions and their tendency to look realistic.
As for the second condition, we perform the following computation: We degrade the obtained reconstruction $\mathbf{\hat{x}}$ by $\mathbf{H}$, and calculate its difference from the given measurement $\mathbf{y}$, obtaining $\mathbf{y} - \mathbf{H \hat{x}}$. According to the problem setting, this difference should be an additive white Gaussian noise vector with a standard deviation of $\sigma_0$. We examine this difference by calculating its empirical standard deviation, and performing the Pearson-D'Agostino~\cite{dagostino} test of normality on it, accepting it as a Gaussian vector if the obtained p-value is greater than $0.05$. We also calculate the Pearson correlation coefficient (denoted as $\rho$) among neighboring entries, accepting them as uncorrelated for coefficients smaller than $0.1$ in absolute value. In all of our tests, the standard deviation matches $\sigma_0$ almost exactly, the Pearson correlation coefficient satisfies $\left|\rho\right| < 0.1$, and we obtain p-values greater than $0.05$ in around $95\%$ of the samples (across all experiments).
These results empirically show that our algorithm produces valid solutions to the given inverse problems.
\section{Conclusion and Future Work}
\label{sec:conclusion}
\begin{figure}
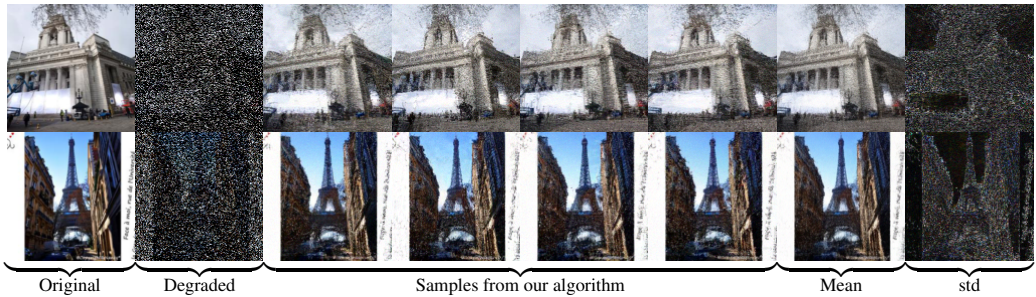

    \centering
    \def\arraystretch{0.1}
    \setlength\tabcolsep{0pt}
    \begin{tabular}{cccccccc}
        \forloop{row}{4}{\value{row} < 6}{
            \hspace{-0.05cm}\includegraphics[width=1.7cm,height=1.7cm]{./images/tower_cs4/sample_\arabic{row}_column_0.png} \hspace{-0.17cm}
            \forloop{col}{1}{\value{col} < 8}{
                & \includegraphics[width=1.7cm,height=1.7cm]{./images/tower_cs4/sample_\arabic{row}_column_\arabic{col}.png} \hspace{-0.17cm}
            } \\
        }
        & \multicolumn{7}{c}{\vspace{0.5mm}}\\
        \upbracefill & \upbracefill &
        \multicolumn{4}{c}{
            \upbracefill
        } &
        \upbracefill & \upbracefill
        \\
        & \multicolumn{7}{c}{\vspace{0.5mm}}\\
        \scriptsize{Original} & \scriptsize{Degraded} &
        \multicolumn{4}{c}{
            \scriptsize{Samples from our algorithm}
        } &
        \scriptsize{Mean} & \scriptsize{std}
    \end{tabular}
    \caption{Compressive sensing results on LSUN~\cite{yu2015lsun} tower images (compression by $25\%$ and adding noise with $\sigma_0=0.04$).}
    \label{fig:cs4_tower}
\end{figure}
SNIPS, presented in this paper, is a novel stochastic algorithm for solving general noisy linear inverse problems. This method is based on annealed Langevin dynamics and Newton’s method, and relies on the availability of a pre-trained Gaussian MMSE denoiser. SNIPS produces a random variety of high quality samples from the posterior distribution of the unknown given the measurements, while guaranteeing their validity with respect to the given data. This algorithm's derivation includes an intricate choice of the injected annealed noise in the Langevin update equations, and an SVD decomposition of the degradation operator for decoupling the measurements' dependencies. We demonstrate SNIPS' success on image deblurring, super resolution, and compressive sensing.

Extensions of this work should focus on SNIPS' limitations: (i) The need to deploy SVD decomposition of the degradation matrix requires a considerable amount of memory and computations, and hinders the algorithm’s scalability; (ii) The current version of SNIPS does not handle general content images, a fact that is related to the properties of the denoiser being used~\cite{ryu2019plug}; and (iii) SNIPS, as any other Langevin based method, requires (too) many iterations (\eg, in our super-resolution tests on CelebA, $2$ minutes are required for producing $8$ sample images), and means for its acceleration should be explored.

\bibliography{refs}

%%%%%%%%%%%%%%%%%%%%%%%%%%%%%%%%%%%%%%%%%%%%%%%%%%%%%%%%%%%%

\appendix
\vfil 
\pagebreak
\section{Conditional Score Derivation Proofs}
\label{sec:proofs}
We would like to derive a term for $\nabla{\mathbf{\tilde{x}}_T} \log p\left(\mathbf{\tilde{x}}_T|\mathbf{y}_T\right)$ depending on known ingredients such as $\mathbf{\tilde{x}}_T$, $\mathbf{y}_T$, $\sigma_0$, $\sigma_i$ and the SVD components of $\mathbf{H}$, as well as the blurred prior score function $\nabla{\mathbf{\tilde{x}}_T} \log p\left(\mathbf{\tilde{x}}_T\right)$, which can be estimated using a neural network.
To that end, in accordance with the definitions of $\mathbf{v}_0$, $\mathbf{v}_<$ and $\mathbf{v}_>$, for a matrix $\mathbf{M}$ we define $\mathbf{M}_0, \mathbf{M}_<, \mathbf{M}_>$ as leading minors of $\mathbf{M}$ with subsets of rows and columns extracted accordingly from the above-defined partition.
Recalling \autoref{eqn:measurements}, we have
\begin{equation}
    \label{eqn:replace}
     \mathbf{y}_T - \mathbf{\Sigma} \mathbf{\tilde{x}}_T =  \mathbf{U}^T \mathbf{z} -\mathbf{\Sigma V}^T \mathbf{n}.
\end{equation}
Observe that the entries of the right-hand-side vector are statistically independent, and their distribution for $s_j < \sigma_0 / \sigma_i$ is given by
\begin{equation}
    \label{eqn:z_small}
    {\left( \mathbf{U}^T \mathbf{z} -\mathbf{\Sigma V}^T \mathbf{n} \right)_{<} \sim \mathcal{N}\left(0, \sigma_0^2 \mathbf{I} - \sigma_i^2 \mathbf{\Sigma}_< \mathbf{\Sigma}_<^T\right)}.
\end{equation}
This is a direct result of \autoref{eqn:relation_to_z}, obtained by simply aggregating the different entries $j$ into a vector.
Similarly,
\begin{equation}
\label{eqn:z_big}
{\left(\mathbf{V}^T \mathbf{n} - \mathbf{\Sigma}^{\dagger} \mathbf{U}^T \mathbf{z}\right)_{\mathbf{>}} \sim \mathcal{N}\left(0, \sigma_i^2 \mathbf{I} - \sigma_0^2 \mathbf{\Sigma}_>^{-1} \mathbf{\Sigma}_>^{-1^T}\right)},
\end{equation}
obtained from aggregating the entries from \autoref{eqn:relation_to_z} into a vector, and multiplying it by $\mathbf{\Sigma}_>^{-1}$. Notice that $\mathbf{\Sigma}_>$ is a diagonal square matrix, and thus invertible. The above two formulae will be used in the following analysis.

\renewcommand{\qedsymbol}{$\blacksquare$}
\newtheorem{theorem}{Theorem}

\begin{theorem}
Given $\mathbf{y} = \mathbf{Hx} + \mathbf{z}$, $\mathbf{z} \sim \mathcal{N}\left(0, \sigma_0^2 \mathbf{I}\right)$, $\mathbf{H} = \mathbf{U \Sigma V}^T$ is the SVD decomposition of $\mathbf{H}$, $\mathbf{y}_T = \mathbf{U}^T \mathbf{y}$,
$\mathbf{n} = \mathbf{n}_i$ as constructed in \autoref{sec:problem},
$\mathbf{\tilde{x}} = \mathbf{\tilde{x}}_i = \mathbf{x} + \mathbf{n}$, $\mathbf{\tilde{x}}_T = \mathbf{V}^T \mathbf{\tilde{x}}$, $\mathbf{{x}}_T = \mathbf{V}^T \mathbf{{x}}$, the conditional score is approximately given by:
\[
    \nabla_{ \mathbf{\tilde{x}}_T } \log p\left(\mathbf{\tilde{x}}_T | \mathbf{y}_T\right) =
    \mathbf{\Sigma}^T \left| \sigma_0^2 \mathbf{I} - \sigma_i^2 \mathbf{\Sigma \Sigma}^T \right|^{\dagger}
    \left( \mathbf{y}_T - \mathbf{\Sigma}  \mathbf{\tilde{x}}_T \right) +
    \left. \left( \mathbf{V}^T \nabla_{\mathbf{\tilde{x}}} \log p\left(\mathbf{\tilde{x}}\right) \right) \right|_{\not{>}}
\]
\end{theorem}
\begin{proof}
We split our derivation into three cases: $\mathbf{\tilde{x}}_{T,0}$, $\mathbf{\tilde{x}}_{T,>}$, and $\mathbf{\tilde{x}}_{T,<}$, and then concatenate the results.

\textbf{For the case of $\mathbf{\tilde{x}}_{T,0}$}, we calculate using the Bayes rule:
\[
    \nabla_{\mathbf{\tilde{x}}_{T,0}} \log p\left(\mathbf{\tilde{x}}_T | \mathbf{y}_T\right) =
    \nabla_{\mathbf{\tilde{x}}_{T,0}} \log p\left(\mathbf{y}_T | \mathbf{\tilde{x}}_T\right) + \nabla_{\mathbf{\tilde{x}}_{T,0}} \log p\left(\mathbf{\tilde{x}}_T\right).
\]
Deriving by $\mathbf{\tilde{x}}_{T,0}$ is the same as deriving by $\mathbf{\tilde{x}}_{T}$ and then taking the part referring to zero singular values of $\mathbf{H}$. Thus, the second term becomes $\left( \nabla_{\mathbf{\tilde{x}}_T } \log p\left(\mathbf{\tilde{x}}_T \right) \right)_{0}$. As for the first term, we can subtract the vector $\mathbf{\Sigma} \mathbf{\tilde{x}}_T$ without changing the statistics because it is a known quantity in this setting, resulting in
\begin{align*}
     \nabla_{\mathbf{\tilde{x}}_{T,0}} \log p\left(\mathbf{\tilde{x}}_T | \mathbf{y}_T\right) & =
    \nabla_{\mathbf{\tilde{x}}_{T,0}} \log p\left(\mathbf{y}_T - \mathbf{\Sigma} \mathbf{\tilde{x}}_T  | \mathbf{\tilde{x}}_T \right) + \left( \nabla_{\mathbf{\tilde{x}}_T } \log p\left(\mathbf{\tilde{x}}_T \right) \right)_{0} \\
    & =
    \nabla_{\mathbf{\tilde{x}}_{T,0}} \log p\left(\mathbf{U}^{T} \mathbf{z} - \mathbf{\Sigma} \mathbf{V}^T \mathbf{n} | \mathbf{\tilde{x}}_T\right) + \left( \nabla_{\mathbf{\tilde{x}}_T} \log p\left(\mathbf{\tilde{x}}_T\right) \right)_{0}.
\end{align*}
The last equality holds due to \autoref{eqn:replace}. Referring to the first term, because the entries of the vector are independent, we can split the probability density function into a product of two such functions for two parts of the vector, as follows:
\[
    \nabla_{\mathbf{\tilde{x}}_{T,0}} \left[ \log p\left(\left(\mathbf{U}^{T} \mathbf{z} - \mathbf{\Sigma V}^T \mathbf{n}\right)_0 | \mathbf{\tilde{x}}_T\right) + \log p\left(\left(\mathbf{U}^{T} \mathbf{z} - \mathbf{\Sigma V}^T \mathbf{n}\right)_{\not{0}} | \mathbf{\tilde{x}}_T\right) \right].
\]
The entries of $\left(\mathbf{U}^{T} \mathbf{z} - \mathbf{\Sigma V}^T \mathbf{n}\right)_{\not{0}}$ were defined element-wise as gradual noise additions, statistically independent of the entries of $\mathbf{\tilde{x}}_{T,0}$. Therefore, the conditioning on $\mathbf{\tilde{x}}_T$ is equivalent to conditioning on $\mathbf{\tilde{x}}_{T,\not{0}}$. Deriving this log-probability by $\mathbf{\tilde{x}}_{T,0}$ results in zero.
As for the first term, $\left(\mathbf{\Sigma V}^T \mathbf{n}\right)_{0}$ is zero due to the definition of $\mathbf{\Sigma}$, and $\left(\mathbf{U}^T \mathbf{z}\right)_{0}$ is a Gaussian vector that is independent of $\mathbf{\tilde{x}}_{T,0}$, resulting in
\begin{align*}
    \nabla_{\mathbf{\tilde{x}}_{T,0}} \log p\left(\mathbf{\tilde{x}}_T | \mathbf{y}_T\right) & =
    \nabla_{\mathbf{\tilde{x}}_{T,0}} \log p\left(\left(\mathbf{U}^{T} \mathbf{z} - \mathbf{\Sigma V}^T \mathbf{n}\right)_{0} | \mathbf{\tilde{x}}_T\right) + \left( \nabla_{\mathbf{\tilde{x}}_T} \log p\left(\mathbf{\tilde{x}}_T\right) \right)_0 \\
& =
    \nabla_{\mathbf{\tilde{x}}_{T,0}} \log p\left(\left(\mathbf{U}^{T} \mathbf{z}\right)_{0} | \mathbf{\tilde{x}}_T\right) + \left( \nabla_{\mathbf{\tilde{x}}_T} \log p\left(\mathbf{\tilde{x}}_T\right) \right)_0 \\
& =
    \left( \nabla_{\mathbf{\tilde{x}}_T} \log p\left(\mathbf{\tilde{x}}_T\right) \right)_0 \\
& =
    \left(\nabla_{\mathbf{\tilde{x}}_T} \log p\left(\mathbf{\tilde{x}}\right) \right)_0 \\
& =
    \left(\mathbf{V}^T \nabla_{\mathbf{\tilde{x}}} \log p\left(\mathbf{\tilde{x}}\right) \right)_0.
\end{align*}
The second last equality holds because $\mathbf{\tilde{x}} = \mathbf{V} \mathbf{\tilde{x}}_T$, and multiplication by the orthogonal matrix $\mathbf{V}$ does not change the statistics of the variable.
The last equality holds due to the multivariate chain rule: $\nabla_{\mathbf{x}} f\left(\mathbf{y}\right) = \mathbf{J}\left(\mathbf{y}\left(\mathbf{x}\right)\right) \nabla_{\mathbf{y}} f\left(\mathbf{y}\right)$, where $\mathbf{J}\left(\mathbf{y}\left(\mathbf{x}\right)\right)$ is the Jacobian matrix of $\mathbf{y}$ w.r.t. $\mathbf{x}$. Finally, we obtain
\begin{equation}
    \label{eqn:apdx_zero}
    \nabla_{\mathbf{\tilde{x}}_{T,0}} \log p\left(\mathbf{\tilde{x}}_T | \mathbf{y}_T\right) =
    \left(\mathbf{V}^T \nabla_{\mathbf{\tilde{x}}} \log p\left(\mathbf{\tilde{x}}\right) \right)_0.
\end{equation}

\textbf{For the case of $\mathbf{\tilde{x}}_{T,>}$}, using the definition of the conditional distribution we get:
\begin{align}
\label{eqn:big_split}
\begin{split}
    \nabla_{\mathbf{\tilde{x}}_{T,>}} \log p\left(\mathbf{\tilde{x}}_T | \mathbf{y}_T\right) & =
    \nabla_{\mathbf{\tilde{x}}_{T,>}} \log p\left(\mathbf{\tilde{x}}_{T,0}, \mathbf{\tilde{x}}_{T,\not{0}} | \mathbf{y}_T\right)\\
& =
    \nabla_{\mathbf{\tilde{x}}_{T,>}} \log p\left(\mathbf{\tilde{x}}_{T,0} | \mathbf{\tilde{x}}_{T,\not{0}}, \mathbf{y}_T\right) +
    \nabla_{\mathbf{\tilde{x}}_{T,>}} \log p\left(\mathbf{\tilde{x}}_{T,\not{0}} | \mathbf{y}_T\right).
\end{split}
\end{align}
Focusing on the second term, we calculate, with a similar reasoning as above and get:
\[
    \nabla_{\mathbf{\tilde{x}}_{T,>}} \log p\left(\mathbf{\tilde{x}}_{T,\not{0}} | \mathbf{y}_T\right) =
    \nabla_{\mathbf{\tilde{x}}_{T,>}} \log p\left(\left(\mathbf{\tilde{x}}_T - \mathbf{\Sigma}^{\dagger} \mathbf{y}_T\right)_{\not{0}} | \mathbf{y}_T\right).
\]
Substituting $\mathbf{\tilde{x}}_T = \mathbf{V}^T \mathbf{x} + \mathbf{V}^T \mathbf{n}$, $\mathbf{y}_T = \mathbf{U}^T \mathbf{H x} + \mathbf{U}^T \mathbf{z}$, $\mathbf{H} = \mathbf{U \Sigma V}^T$ leads to
\begin{align*}
    \nabla_{\mathbf{\tilde{x}}_{T,>}} \log p\left(\mathbf{\tilde{x}}_{T, \not{0}} | \mathbf{y}_T\right) & =
    \nabla_{\mathbf{\tilde{x}}_{T,>}} \log p\left(\left(\mathbf{V}^T \mathbf{x} + \mathbf{V}^T \mathbf{n} - \mathbf{\Sigma}^{\dagger} \left( \mathbf{U}^T \mathbf{H} \mathbf{x} + \mathbf{U}^T \mathbf{z} \right)\right)_{\not{0}} | \mathbf{y}_T\right) \\
& =
    \nabla_{\mathbf{\tilde{x}}_{T,>}} \log p\left(\left(\mathbf{V}^T \mathbf{x} + \mathbf{V}^T \mathbf{n} - \mathbf{\Sigma}^{\dagger} \mathbf{U}^T \mathbf{U \Sigma V}^T \mathbf{x} - \mathbf{\Sigma}^{\dagger} \mathbf{U}^T \mathbf{z}\right)_{\not{0}} | \mathbf{y}_T\right) \\
& =
    \nabla_{\mathbf{\tilde{x}}_{T,>}} \log p\left(\left(\mathbf{V}^T \mathbf{n} - \mathbf{\Sigma}^{\dagger} \mathbf{U}^T \mathbf{z} + \left( \mathbf{I} - \mathbf{\Sigma}^\dagger \mathbf{\Sigma} \right) \mathbf{x}_T\right)_{\not{0}} | \mathbf{y}_T\right).
\end{align*}
The last equality holds because $\mathbf{U}^T \mathbf{U} = \mathbf{I}$. Observe that $\left( \mathbf{I} - \mathbf{\Sigma}^\dagger \mathbf{\Sigma} \right) \mathbf{x}_T$ is zero everywhere except in the $0$ part of the vector, which we discard because of the $\not{0}$ notation.
We can split this term into two parts, as before,
\[
    \nabla_{\mathbf{\tilde{x}}_{T,>}} \left[ \log p\left(\left(\mathbf{V}^T \mathbf{n} - \mathbf{\Sigma}^{\dagger} \mathbf{U}^T \mathbf{z}\right)_{>} | \mathbf{y}_T\right)
    +  \log p\left(\left(\mathbf{V}^T \mathbf{n} - \mathbf{\Sigma}^{\dagger} \mathbf{U}^T \mathbf{z}\right)_{<} | \mathbf{y}_T\right)
    \right].
\]
The derivative of the second term (the $<$ part) is zero, because this vector was built element-wise as gradual noise additions, independent of $\mathbf{\tilde{x}}_{T,>}$. This results in
\begin{align*}
    \nabla_{\mathbf{\tilde{x}}_{T,>}} \log p\left(\mathbf{\tilde{x}}_{T, \not{0}} | \mathbf{y}_T\right)
    & = \nabla_{\mathbf{\tilde{x}}_{T,>}} \log p\left(\left(\mathbf{V}^T \mathbf{n} - \mathbf{\Sigma}^{\dagger} \mathbf{U}^T \mathbf{z}\right)_{>} | \mathbf{y}_T\right) \\
& =
    \nabla_{\mathbf{\tilde{x}}_{T,>}} \log p\left(\left(\mathbf{\tilde{x}}_T - \mathbf{\Sigma}^{\dagger} \mathbf{y}_T\right)_{>} | \mathbf{y}_T\right)
\end{align*}
This is the gradient-log of a Gaussian density function of the vector ${\left(\mathbf{\tilde{x}}_T - \mathbf{\Sigma}^{\dagger} \mathbf{y}_T\right)_>}$, known to have a zero mean and a covariance matrix ${\sigma_i^2 \mathbf{I} - \sigma_0^2 \mathbf{\Sigma}_>^{-1} \mathbf{\Sigma}_>^{-1^T}}$, according to \autoref{eqn:z_big}. Thus, we use the known Gaussian gradient-log and conclude:
\begin{align*}
    \nabla_{\mathbf{\tilde{x}}_{T,>}} \log p\left(\mathbf{\tilde{x}}_{T, \not{0}} | \mathbf{y}_T\right) & =
    \left( \sigma_i^2 \mathbf{I} - \sigma_0^2 \mathbf{\Sigma}_>^{-1} \mathbf{\Sigma}_>^{-1^T} \right)^{-1} \left(\mathbf{\Sigma}^{\dagger} \mathbf{y}_T - \mathbf{\tilde{x}}_T\right)_{>} \\
& =
    \left( \mathbf{\Sigma}_>^{-1} \left( \mathbf{\Sigma}_> \sigma_i^2 \mathbf{I} \mathbf{\Sigma}_>^{T} - \sigma_0^2 \mathbf{I} \right) \mathbf{\Sigma}_>^{-1^T} \right)^{-1}
    \left(\mathbf{\Sigma}^{\dagger} \mathbf{y}_T - \mathbf{\tilde{x}}_T\right)_{>} \\
& =
    \mathbf{\Sigma}_>^{T}
    \left( \sigma_i^2 \mathbf{\Sigma}_> \mathbf{\Sigma}_>^{T} - \sigma_0^2 \mathbf{I} \right)^{-1}
    \mathbf{\Sigma}_> \left(\mathbf{\Sigma}^{\dagger} \mathbf{y}_T - \mathbf{\tilde{x}}_T\right)_{>}.
\end{align*}
Multiplying a certain part of a diagonal matrix (in this case, the $>$ part) by the corresponding part of a vector is the same as multiplying the original matrix and vector, and then taking the relevant part. This results in
\begin{align}
\label{eqn:big_1st}
\begin{split}
    \nabla_{\mathbf{\tilde{x}}_{T,>}} \log p\left(\mathbf{\tilde{x}}_{T, \not{0}} | \mathbf{y}_T\right) & =
    \mathbf{\Sigma}_>^{T}
    \left( \sigma_i^2 \mathbf{\Sigma}_> \mathbf{\Sigma}_>^{T} - \sigma_0^2 \mathbf{I} \right)^{-1}
    \left(\mathbf{\Sigma} \mathbf{\Sigma}^{\dagger} \mathbf{y}_T - \mathbf{\Sigma} \mathbf{\tilde{x}}_T\right)_{>} \\
& = 
    \mathbf{\Sigma}_>^{T}
    \left( \sigma_i^2 \mathbf{\Sigma}_> \mathbf{\Sigma}_>^{T} - \sigma_0^2 \mathbf{I} \right)^{-1}
    \left(\mathbf{y}_T - \mathbf{\Sigma} \mathbf{\tilde{x}}_T\right)_{>} \\
& =
    \left( \mathbf{\Sigma}^{T}
    \left( \sigma_i^2 \mathbf{\Sigma \Sigma}^T - \sigma_0^2 \mathbf{I} \right)^{-1}
    \left(\mathbf{y}_T - \mathbf{\Sigma} \mathbf{\tilde{x}}_T\right) \right)_{>}.
\end{split}
\end{align}
As for the first term in \autoref{eqn:big_split}, which is $\nabla_{\mathbf{\tilde{x}}_{T,>}} \log p\left(\mathbf{\tilde{x}}_{T,0} | \mathbf{\tilde{x}}_{T,\not{0}}, \mathbf{y}_T\right)$,
we can rewrite it as $\nabla_{\mathbf{\tilde{x}}_{T,>}} \log p\left(\mathbf{\tilde{x}}_{T,0} | \mathbf{\tilde{x}}_{T,\not{0}}, \mathbf{\Sigma}_{\not{0}}^{-1} \mathbf{y}_T\right)$ because $\mathbf{\Sigma}_{\not{0}}^{-1}$ is an orthogonal matrix that does not add or remove information.
Furthermore, we notice that the difference $\mathbf{\tilde{x}}_{T,\not{0}} - \mathbf{\Sigma}_{\not{0}}^{-1} \mathbf{y}_T$ was defined element-wise as gradual noise additions, independent of $\mathbf{\tilde{x}}_{T,0}$.
Therefore, the term can be rewritten as $\nabla_{\mathbf{\tilde{x}}_{T,>}} \log p\left(\mathbf{\tilde{x}}_{T,0} | \mathbf{\tilde{x}}_{T,\not{0}}\right)$. We calculate using the definition of the conditional distribution:
\begin{align*}
    \nabla_{\mathbf{\tilde{x}}_{T,>}} \log p\left(\mathbf{\tilde{x}}_{T,0} | \mathbf{\tilde{x}}_{T,\not{0}}\right) & =
    \nabla_{\mathbf{\tilde{x}}_{T,>}} \log \frac{p\left(\mathbf{\tilde{x}}_{T,0}, \mathbf{\tilde{x}}_{T,\not{0}}\right)}{p\left(\mathbf{\tilde{x}}_{T,\not{0}}\right)} \\
& =
    \nabla_{\mathbf{\tilde{x}}_{T,>}} \log p\left(\mathbf{\tilde{x}}_{T,0}, \mathbf{\tilde{x}}_{T,\not{0}}\right) -
    \nabla_{\mathbf{\tilde{x}}_{T,>}} \log p\left(\mathbf{\tilde{x}}_{T,\not{0}}\right) \\
& =
    \nabla_{\mathbf{\tilde{x}}_{T,>}} \log p\left(\mathbf{\tilde{x}}_T\right) -
    \nabla_{\mathbf{\tilde{x}}_{T,>}} \log p\left(\mathbf{\tilde{x}}_{T,\not{0}}\right)  \\
& =
    \left(\nabla_{\mathbf{\tilde{x}}_T} \log p\left(\mathbf{\tilde{x}}_T\right)\right)_> -
    \left( \nabla_{\mathbf{\tilde{x}}_{T,\not{0}}} \log p\left(\mathbf{\tilde{x}}_{T,\not{0}}\right)\right)_{>}\\
& =
    \left(\mathbf{V}^T \nabla_{\mathbf{\tilde{x}}} \log p\left(\mathbf{\tilde{x}}\right)\right)_> -
    \left(\mathbf{V}_{\not{0}}^T \nabla_{\mathbf{\tilde{x}}_{\not{0}}} \log p\left(\mathbf{\tilde{x}}_{\not{0}}\right)\right)_{>}\\
& =
    \mathbf{V}_>^T \left(\nabla_{\mathbf{\tilde{x}}} \log p\left(\mathbf{\tilde{x}}\right)\right)_> -
    \mathbf{V}_>^T  \left(\nabla_{\mathbf{\tilde{x}}_{\not{0}}} \log p\left(\mathbf{\tilde{x}}_{\not{0}}\right)\right)_{>}.
\end{align*}
The second last equality holds due to the chain rule, and the last one holds because multiplying the $>$ part of a diagonal matrix by the corresponding part of a vector is the same as multiplying the original matrix and vector, and then taking the relevant part, as previously mentioned.
Recalling \autoref{eqn:denoiser}, we can substitute both terms by their denoiser counterparts, obtaining
\begin{align*}
    \nabla_{\mathbf{\tilde{x}}_{T,>}} \log p\left(\mathbf{\tilde{x}}_{T,0} | \mathbf{\tilde{x}}_{T,\not{0}}\right) & =
    \mathbf{V}_>^T \left(\frac{\mathbb{E}\left[\mathbf{x} | \mathbf{\tilde{x}}\right] - \mathbf{\tilde{x}}}{\sigma_i^2} \right)_> -
    \mathbf{V}_>^T \left(\frac{\mathbb{E}\left[\mathbf{x}_{\not{0}} | \mathbf{\tilde{x}}_{\not{0}}\right] - \mathbf{\tilde{x}}_{\not{0}}}{\sigma_i^2} \right)_> \\
& =
    \frac{1}{\sigma_i^2} \mathbf{V}_>^T \left( \left(\mathbb{E}\left[\mathbf{x} | \mathbf{\tilde{x}}\right]\right)_> - \mathbf{\tilde{x}}_> - \left(\mathbb{E}\left[\mathbf{x}_{\not{0}} | \mathbf{\tilde{x}}_{\not{0}}\right]\right)_> + \mathbf{\tilde{x}}_{>} \right) \\
& =
    \frac{1}{\sigma_i^2} \mathbf{V}_>^T \left( \left(\mathbb{E}\left[\mathbf{x} | \mathbf{\tilde{x}}\right]\right)_> - \left(\mathbb{E}\left[\mathbf{x}_{\not{0}} | \mathbf{\tilde{x}}_{\not{0}}\right]\right)_> \right) \\
& =
    \frac{1}{\sigma_i^2} \mathbf{V}_>^T \left( \mathbb{E}\left[\mathbf{x}_> | \mathbf{\tilde{x}}\right] - \mathbb{E}\left[\mathbf{x}_> | \mathbf{\tilde{x}}_{\not{0}}\right] \right).
\end{align*}
We obtained a difference betwen two terms, both of which calculate an expectation of $\mathbf{x}_>$ given $\mathbf{\tilde{x}}_{\not{0}}$, with the first term including the extra knowledge of $\mathbf{\tilde{x}}_0$. We introduce an assumption that this additional information does not significantly change the estimation of $\mathbf{x}_>$, especially since a noisy version of it, $\mathbf{\tilde{x}}_>$, is given in both estimators. As a result, we obtain the approximation
\begin{equation}
    \label{eqn:big_2nd}
    \nabla_{\mathbf{\tilde{x}}_{T,>}} \log p\left(\mathbf{\tilde{x}}_{T,0} | \mathbf{\tilde{x}}_{T,\not{0}}\right) \approx \mathbf{0}.
\end{equation}

To conclude this part, we combine Equations \ref{eqn:big_split}, \ref{eqn:big_1st} and \ref{eqn:big_2nd} and obtain the approximate relation
\begin{equation}
    \label{eqn:apdx_big}
    \nabla_{\mathbf{\tilde{x}}_{T,>}} \log p\left(\mathbf{\tilde{x}}_T | \mathbf{y}_T\right) =
    \left( \mathbf{\Sigma}^{T}
    \left( \sigma_i^2 \mathbf{\Sigma \Sigma}^T - \sigma_0^2 \mathbf{I} \right)^{-1}
    \left(\mathbf{y}_T - \mathbf{\Sigma} \mathbf{\tilde{x}}_T\right) \right)_{>}.
\end{equation}

\textbf{For the case of $\mathbf{\tilde{x}}_{T,<}$}, we calculate using the Bayes rule, with similar reasoning to previous cases:
\begin{align*}
    \nabla_{\mathbf{\tilde{x}}_{T,<}} \log p\left(\mathbf{\tilde{x}}_T | \mathbf{y}_T\right) & =
    \nabla_{\mathbf{\tilde{x}}_{T,<}} \log p\left(\mathbf{y}_T | \mathbf{\tilde{x}}_T\right) + \nabla_{\mathbf{\tilde{x}}_{T,<}} \log p\left(\mathbf{\tilde{x}}_T\right) \\
& =
    \nabla_{\mathbf{\tilde{x}}_{T,<}} \log p\left(\mathbf{y}_T - \mathbf{\Sigma} \mathbf{\tilde{x}}_T | \mathbf{\tilde{x}}_T\right) + \left( \nabla_{\mathbf{\tilde{x}}_T} \log p\left(\mathbf{\tilde{x}}_T\right) \right)_< \\
& =
    \nabla_{\mathbf{\tilde{x}}_{T,<}} \log p\left(\mathbf{U}^{T} \mathbf{z} - \mathbf{\Sigma V}^T \mathbf{n} | \mathbf{\tilde{x}}_T\right) + \left( \nabla_{\mathbf{\tilde{x}}_T} \log p\left(\mathbf{\tilde{x}}_T\right) \right)_<.
\end{align*}
Similar to the first case, we can split the first term in the same fashion and obtain
\[
    \nabla_{\mathbf{\tilde{x}}_{T,<}} \left[ \log p\left(\left(\mathbf{U}^{T} \mathbf{z} - \mathbf{\Sigma V}^T \mathbf{n}\right)_< | \mathbf{\tilde{x}}_T\right) +  \log p\left(\left(\mathbf{U}^{T} \mathbf{z} - \mathbf{\Sigma V}^T \mathbf{n}\right)_{\not{<}} | \mathbf{\tilde{x}}_T\right) \right] =
\]
\[ =
    \nabla_{\mathbf{\tilde{x}}_{T,<}} \log p\left(\left(\mathbf{U}^{T} \mathbf{z} - \mathbf{\Sigma V}^T \mathbf{n}\right)_{<} | \mathbf{\tilde{x}}_T\right)
    = \nabla_{\mathbf{\tilde{x}}_{T,<}} \log p\left(\left(\mathbf{y}_T - \mathbf{\Sigma} \mathbf{\tilde{x}}_T\right)_{<} | \mathbf{\tilde{x}}_T\right).
\]
The vector $\left(\mathbf{U}^{T} \mathbf{z} - \mathbf{\Sigma V}^T \mathbf{n}\right)_{\not{<}}$ was built element-wise as gradual noise additions, independent of $\mathbf{\tilde{x}}_{T,<}$, and thus its derivative is zero.
We obtain a gradient-log of a Gaussian density function of the vector ${\left(\mathbf{y}_T - \mathbf{\Sigma} \mathbf{\tilde{x}}_T\right)_<}$, having a zero mean and a covariance matrix ${\sigma_0^2 \mathbf{I} - \sigma_i^2 \mathbf{\Sigma}_< \mathbf{\Sigma}_<^T}$, according to \autoref{eqn:z_small}. Thus, when deriving it by $\mathbf{\tilde{x}}_{T,<}$, we obtain the known Gaussian gradient-log, multiplied from the left by $-\mathbf{\Sigma}_<^T$, which is the inner derivative of the Gaussian parameter, implying
\begin{align*}
    \nabla_{\mathbf{\tilde{x}}_{T,<}} \log p\left(\mathbf{\tilde{x}}_T | \mathbf{y}_T\right) & =
    - \mathbf{\Sigma}_<^T \left( \sigma_0^2 \mathbf{I} - \sigma_i^2 \mathbf{\Sigma}_< \mathbf{\Sigma}_<^T \right)^{-1} \left( \mathbf{\Sigma} \mathbf{\tilde{x}}_T - \mathbf{y}_T \right)_< + \left( \nabla_{\mathbf{\tilde{x}}_T} \log p\left(\mathbf{\tilde{x}}_T\right) \right)_< \\
& =
    \mathbf{\Sigma}_<^T \left( \sigma_0^2 \mathbf{I} - \sigma_i^2 \mathbf{\Sigma}_< \mathbf{\Sigma}_<^T \right)^{-1} \left( \mathbf{y}_T - \mathbf{\Sigma} \mathbf{\tilde{x}}_T \right)_< + \left( \nabla_{\mathbf{\tilde{x}}_T} \log p\left(\mathbf{\tilde{x}}_T\right) \right)_< \\
& =
    \left( \mathbf{\Sigma}^T \left( \sigma_0^2 \mathbf{I} - \sigma_i^2 \mathbf{\Sigma \Sigma}^T \right)^{-1} \left( \mathbf{y}_T - \mathbf{\Sigma} \mathbf{\tilde{x}}_T \right) \right)_< + \left(\mathbf{V}^T \nabla_{\mathbf{\tilde{x}}} \log p\left(\mathbf{\tilde{x}}\right) \right)_<.
\end{align*}
So, in summary,
\begin{equation}
    \label{eqn:apdx_small}
    \nabla_{\mathbf{\tilde{x}}_{T,<}} \log p\left(\mathbf{\tilde{x}}_T | \mathbf{y}_T\right) =
    \left( \mathbf{\Sigma}^T \left( \sigma_0^2 \mathbf{I} - \sigma_i^2 \mathbf{\Sigma \Sigma}^T \right)^{-1} \left( \mathbf{y}_T - \mathbf{\Sigma} \mathbf{\tilde{x}}_T \right) \right)_< + \left(\mathbf{V}^T \nabla_{\mathbf{\tilde{x}}} \log p\left(\mathbf{\tilde{x}}\right) \right)_<.
\end{equation}

Aggregating all these results together, by combining Equations \ref{eqn:apdx_zero}, \ref{eqn:apdx_big} and \ref{eqn:apdx_small} into one vector, we obtain the following conditional score function approximation:
\begin{equation}
\label{eqn:apdx_grad}
    \nabla_{ \mathbf{\tilde{x}}_T } \log p\left(\mathbf{\tilde{x}}_T | \mathbf{y}_T\right) =
    \mathbf{\Sigma}^T \left| \sigma_0^2 \mathbf{I} - \sigma_i^2 \mathbf{\Sigma \Sigma}^T \right|^{\dagger}
    \left( \mathbf{y}_T - \mathbf{\Sigma}  \mathbf{\tilde{x}}_T \right) +
    \left. \left( \mathbf{V}^T \nabla_{\mathbf{\tilde{x}}} \log p\left(\mathbf{\tilde{x}}\right) \right) \right|_{\not{>}},
\end{equation}
where $\left. \left( \mathbf{v} \right) \right|_{\not{>}}$ is the vector $\mathbf{v}$, but with zeros in its entries that correspond to $\mathbf{v}_{>}$. Observe that the first term in \autoref{eqn:apdx_grad} contains zeros in the entries corresponding to $\mathbf{\tilde{x}}_{T, 0}$, matching the above calculations.

\end{proof}

\section{Step Size Derivation}
\label{sec:proof_step}
As explained in~\cite{simoncelli}, the following equality holds:
\[
    \nabla_{\mathbf{\tilde{x}}} \log p\left(\mathbf{\tilde{x}}\right) = \frac{\mathbf{D}\left(\mathbf{\tilde{x}}, \sigma\right) - \mathbf{\tilde{x}}}{\sigma^2},
\]
where $\mathbf{D}\left(\mathbf{\tilde{x}}, \sigma\right)$ is the theoretical MSE minimizer, $\mathbb{E}\left[\mathbf{x} | \mathbf{\tilde{x}}\right]$. We introduce an assumption that $\mathbf{D}\left(\mathbf{\tilde{x}}, \sigma\right)$ does not significantly change with small perturbations in $\mathbf{\tilde{x}}$, resulting in:
\[
    \frac{\partial}{\partial \mathbf{\tilde{x}}} \mathbf{D}\left(\mathbf{\tilde{x}}, \sigma\right) \approx \mathbf{0}.
\]
This assumption is justified by the fact that with probability $1$, the infinitesimal perturbations are orthogonal to the image manifold around the point $\mathbf{\tilde{x}}$, implying that they can be referred to as an additive white Gaussian noise. Due to the efficiency of the denoiser in wiping such noise, the sensitivity of its output to this extra noise is negligible. 

Our goal in this appendix is to evaluate the Hessian of the log posterior in order to be used for better conditioning of the iterative Langevin steps. Thus, we need to differentiate the gradient that was derived above,
\[
    \nabla_{ \mathbf{\tilde{x}}_T } \log p\left(\mathbf{\tilde{x}}_T | \mathbf{y}_T\right) =
    \mathbf{\Sigma}^T \left| \sigma_0^2 \mathbf{I} - \sigma_i^2 \mathbf{\Sigma \Sigma}^T \right|^{\dagger}
    \left( \mathbf{y}_T - \mathbf{\Sigma}  \mathbf{\tilde{x}}_T \right) +
    \left. \left( \mathbf{V}^T \nabla_{\mathbf{\tilde{x}}} \log p\left(\mathbf{\tilde{x}}\right) \right) \right|_{\not{>}}.
\]

\begin{theorem}
Given $\mathbf{y} = \mathbf{Hx} + \mathbf{z}$, $\mathbf{z} \sim \mathcal{N}\left(0, \sigma_0^2 \mathbf{I}\right)$, $\mathbf{H} = \mathbf{U \Sigma V}^T$ is the SVD decomposition of $\mathbf{H}$, $\mathbf{y}_T = \mathbf{U}^T \mathbf{y}$,
$\mathbf{n} = \mathbf{n}_i$ as constructed in \autoref{sec:problem},
$\mathbf{\tilde{x}} = \mathbf{\tilde{x}}_i = \mathbf{x} + \mathbf{n}$, $\mathbf{\tilde{x}}_T = \mathbf{V}^T \mathbf{\tilde{x}}$, $\mathbf{{x}}_T = \mathbf{V}^T \mathbf{{x}}$, the Hessian of the log posterior can be approximated by a diagonal matrix whose entries are:
\[
    \left[ \nabla_{\mathbf{\tilde{x}}_{T}}^2 \log p\left(\mathbf{\tilde{x}}_T | \mathbf{y}_T\right)\right]_{j,j} = \begin{cases}
    \frac{-1}{\sigma_i^2}
    & s_j = 0 \\
    \frac{- s_j^2}{s_j^2 \sigma_i^2 - \sigma_0^2}
    & \sigma_{i} s_j > \sigma_0\\
    \frac{- s_j^2}{\sigma_0^2 - s_j^2 \sigma_i^2} - \frac{1}{\sigma_i^2}
    & 0 < \sigma_{i} s_j < \sigma_0.
    \end{cases}
\]
\end{theorem}
\begin{proof}
Again, we split our calculation into 3 cases:

\textbf{For the case of $\mathbf{\tilde{x}}_{T,0}$}, we notice that the first term in the gradient is zero  due to the multiplication by $\mathbf{\Sigma}^T$, and thus we calculate:
\[
    \nabla_{\mathbf{\tilde{x}}_{T,0}}^2 \log p\left(\mathbf{\tilde{x}}_T | \mathbf{y}_T\right) =
    \frac{\partial}{\partial \mathbf{\tilde{x}}_{T,0}}
    \left. \left( \mathbf{V}^T \nabla_{\mathbf{\tilde{x}}} \log p\left(\mathbf{\tilde{x}}\right) \right) \right|_{\not{>}}.
\]
We use the chain rule and obtain
\begin{align*}
    \nabla_{\mathbf{\tilde{x}}_{T,0}}^2 \log p\left(\mathbf{\tilde{x}}_T | \mathbf{y}_T\right) & =
    \mathbf{V}_0 \frac{\partial}{\partial \mathbf{\tilde{x}}_0}
    \left. \left( \mathbf{V}^T \nabla_{\mathbf{\tilde{x}}} \log p\left(\mathbf{\tilde{x}}\right) \right) \right|_{\not{>}} \\
& =
    \mathbf{V}_0 \mathbf{V}_0^T \frac{\partial}{\partial \mathbf{\tilde{x}}_0}
    \left. \left( \nabla_{\mathbf{\tilde{x}}} \log p\left(\mathbf{\tilde{x}}\right) \right) \right|_{\not{>}} \\
& =
    \frac{\partial}{\partial \mathbf{\tilde{x}}_0}
    \left. \left( \nabla_{\mathbf{\tilde{x}}} \log p\left(\mathbf{\tilde{x}}\right) \right) \right|_{\not{>}} \\
& = 
    \frac{\partial}{\partial \mathbf{\tilde{x}}_0}
    \left. \left( \frac{\mathbf{D}\left(\mathbf{\tilde{x}}, \sigma\right) - \mathbf{\tilde{x}}}{\sigma_i^2} \right) \right|_{\not{>}} \\
& =
    \frac{\partial}{\partial \mathbf{\tilde{x}}_0}
    \left. \left( \frac{ - \mathbf{\tilde{x}}}{\sigma_i^2} \right) \right|_{\not{>}} \\
& =
    \frac{-1}{\sigma_i^2} \mathbf{I},
\end{align*}
where we have invoked our earlier assumption on the denoiser's sensitivity to perturbations. This leads to the conclusion
\begin{equation}
    \label{eqn:stp_zero}
    \nabla_{\mathbf{\tilde{x}}_{T,0}}^2 \log p\left(\mathbf{\tilde{x}}_T | \mathbf{y}_T\right)
    = \frac{-1}{\sigma_i^2} \mathbf{I}.
\end{equation}

\textbf{For the case of $\mathbf{\tilde{x}}_{T,>}$}, we calculate:
\[
    \nabla_{\mathbf{\tilde{x}}_{T,>}}^2 \log p\left(\mathbf{\tilde{x}}_T | \mathbf{y}_T\right) =
    \frac{\partial}{\partial \mathbf{\tilde{x}}_{T,>}} \left[
    \mathbf{\Sigma}^T \left| \sigma_0^2 \mathbf{I} - \sigma_i^2 \mathbf{\Sigma \Sigma}^T \right|^{\dagger}
    \left( \mathbf{y}_T - \mathbf{\Sigma}  \mathbf{\tilde{x}}_T \right) +
    \left. \left( \mathbf{V}^T \nabla_{\mathbf{\tilde{x}}} \log p\left(\mathbf{\tilde{x}}\right) \right) \right|_{\not{>}}
    \right].
\]
The first term's derivative is simply the matrix that multiplies the vector $\mathbf{\tilde{x}}_{T,>}$, while the second term can be approximated, with the use of the chain rule, as follows:
\begin{align*}
    \frac{\partial}{\partial \mathbf{\tilde{x}}_{T,>}}  \left. \left( \mathbf{V}^T \nabla_{\mathbf{\tilde{x}}} \log p\left(\mathbf{\tilde{x}}\right) \right) \right|_{\not{>}}
    & =
    \mathbf{V}_> \frac{\partial}{\partial \mathbf{\tilde{x}}_>}
    \left. \left( \mathbf{V}^T \nabla_{\mathbf{\tilde{x}}} \log p\left(\mathbf{\tilde{x}}\right) \right) \right|_{\not{>}} \\
& =
    \mathbf{V}_> \mathbf{V}_>^T \frac{\partial}{\partial \mathbf{\tilde{x}}_>}
    \left. \left( \nabla_{\mathbf{\tilde{x}}} \log p\left(\mathbf{\tilde{x}}\right) \right) \right|_{\not{>}} \\
& =
    \frac{\partial}{\partial \mathbf{\tilde{x}}_>}
    \left. \left( \nabla_{\mathbf{\tilde{x}}} \log p\left(\mathbf{\tilde{x}}\right) \right) \right|_{\not{>}} \\
& = 
    \frac{\partial}{\partial \mathbf{\tilde{x}}_>}
    \left. \left( \frac{\mathbf{D}\left(\mathbf{\tilde{x}}, \sigma\right) - \mathbf{\tilde{x}}}{\sigma_i^2} \right) \right|_{\not{>}} \\
& =
    \frac{\partial}{\partial \mathbf{\tilde{x}}_>}
    \left. \left( \frac{ - \mathbf{\tilde{x}}}{\sigma_i^2} \right) \right|_{\not{>}} \\
& =
    \mathbf{0},
\end{align*}
where we have invoked our earlier assumption on the denoiser's sensitivity to perturbations, resulting in
\begin{equation}
    \label{eqn:stp_big}
    \nabla_{\mathbf{\tilde{x}}_{T,>}}^2 \log p\left(\mathbf{\tilde{x}}_T | \mathbf{y}_T\right) =
    \left( - \mathbf{\Sigma}^{T} \left| \sigma_0^2 \mathbf{I} - \sigma_i^2 \mathbf{\Sigma \Sigma}^T \right|^{\dagger} \mathbf{\Sigma} \right)_{>}.
\end{equation}

\textbf{For the case of $\mathbf{\tilde{x}}_{T,<}$}, we calculate:
\[
    \nabla_{\mathbf{\tilde{x}}_{T,<}}^2 \log p\left(\mathbf{\tilde{x}}_T | \mathbf{y}_T\right) =
    \frac{\partial}{\partial \mathbf{\tilde{x}}_{T,<}} \left[
    \mathbf{\Sigma}^T \left| \sigma_0^2 \mathbf{I} - \sigma_i^2 \mathbf{\Sigma \Sigma}^T \right|^{\dagger}
    \left( \mathbf{y}_T - \mathbf{\Sigma}  \mathbf{\tilde{x}}_T \right) +
    \left. \left( \mathbf{V}^T \nabla_{\mathbf{\tilde{x}}} \log p\left(\mathbf{\tilde{x}}\right) \right) \right|_{\not{>}}
    \right].
\]
The first term's derivative can be calculated similarly to the previous case, and the second term can be approximately derived as in the first case, resulting in
\begin{equation}
    \label{eqn:stp_small}
    \nabla_{\mathbf{\tilde{x}}_{T,<}}^2 \log p\left(\mathbf{\tilde{x}}_T | \mathbf{y}_T\right) =
    \left( - \mathbf{\Sigma}^{T} \left| \sigma_0^2 \mathbf{I} - \sigma_i^2 \mathbf{\Sigma \Sigma}^T \right|^{\dagger} \mathbf{\Sigma} \right)_< +
    \frac{-1}{\sigma_i^2} \mathbf{I}.
\end{equation}

Aggregating all these results together, by combining Equations \ref{eqn:stp_zero}, \ref{eqn:stp_big} and \ref{eqn:stp_small} into one diagonal matrix, we obtain the following diagonal entries of the Hessian:
\begin{equation}
    \left[ \nabla_{\mathbf{\tilde{x}}_{T}}^2 \log p\left(\mathbf{\tilde{x}}_T | \mathbf{y}_T\right)\right]_{j,j} = \begin{cases}
    \frac{-1}{\sigma_i^2}
    & s_j = 0 \\
    \frac{- s_j^2}{s_j^2 \sigma_i^2 - \sigma_0^2}
    & \sigma_{i} s_j > \sigma_0 \\
    \frac{- s_j^2}{\sigma_0^2 - s_j^2 \sigma_i^2} - \frac{1}{\sigma_i^2}
    & 0 < \sigma_{i} s_j < \sigma_0.
    \end{cases}
\end{equation}
\end{proof}

Finally, since the approximation of the Hessian is a diagonal matrix and its diagonal entries are non-zeros, we can easily invert it. This results in the following term for each of the diagonal entries of the negative inverse Hessian, which we denote $\mathbf{A}_i$:
\[
    \left(\mathbf{A}_i\right)_{j,j} = \begin{cases}
    \sigma_i^2 & s_j = 0 \\
    \sigma_i^2 - \frac{\sigma_0^2}{s_j^2} & \sigma_{i} s_j > \sigma_0 \\
    \sigma_i^2 \cdot \left( 1 - s_j^2 \frac{\sigma_i^2}{\sigma_0^2} \right) & 0 < \sigma_{i} s_j < \sigma_0.
    \end{cases}
\]

\begin{figure}
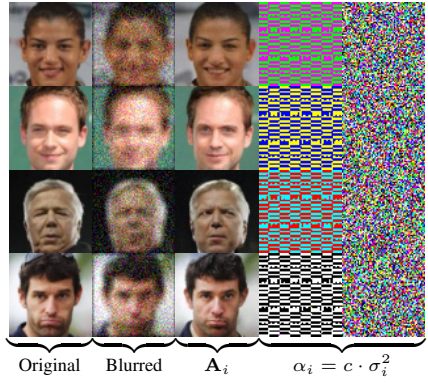

    \centering
    \def\arraystretch{0.1}
    \setlength\tabcolsep{0pt}
    \begin{tabular}{ccccc}
        \forloop{row}{0}{\value{row} < 4}{
            \hspace{-0.05cm}\includegraphics[width=1.1cm,height=1.1cm]{./images/ablation/sample_\arabic{row}_column_0.png} \hspace{-0.17cm}
            \forloop{col}{1}{\value{col} < 5}{
                & \includegraphics[width=1.1cm,height=1.1cm]{./images/ablation/sample_\arabic{row}_column_\arabic{col}.png} \hspace{-0.17cm}
            } \\
        }
        & \multicolumn{4}{c}{\vspace{0.5mm}}\\
        \upbracefill & \upbracefill &
        \multicolumn{1}{c}{
            \upbracefill
        } &
        \multicolumn{2}{c}{
            \upbracefill
        }
        \\
        & \multicolumn{4}{c}{\vspace{0.5mm}}\\
        \scriptsize{Original} & \scriptsize{Blurred} &
        \multicolumn{1}{c}{
            \scriptsize{$\mathbf{A}_i$}
        } &
        \multicolumn{2}{c}{
            \scriptsize{$\alpha_i = c \cdot \sigma_i^2$}
        }
    \end{tabular}
    \caption{Comparison of different step sizes, while the rest of the hyperparameters are fixed (uniform $5 \times 5$ blur and an additive noise with $\sigma_0 = 0.1$). The third column refers to the diagonal step size matrix $\mathbf{A}_i$, as used in SNIPS. The last two columns refer to a uniform time-dependent step size $\alpha_i = c \cdot \sigma_i^2$, with $c = 1e-3, 1e-5$, respectively. Different choices of $c$ yielded similar results.}
    \label{fig:ablation}
\end{figure}

In order to demonstrate the effectiveness of this position-dependent step size vector, we compare it to a uniform step size $\alpha_i \propto \sigma_i^2$ for image deblurring. As can be seen in \autoref{fig:ablation}, the latter diverges under the same hyperparameters. It is possible that for a large enough number of iterations, a uniform step size might converge and produce viable results. However, we find little value in demonstrating this, as it requires retraining the NCSNv2 model for more noise levels, and it slows down the algorithm.

\section{Alternative Definition of the Noise}
\label{sec:plus_appendix}
In our derivations in \autoref{sec:problem} we argued that for the analysis to go through, we should tie the synthetic annealed Langevin noise to the measurements one. As can be seen in \autoref{sec:proofs}, this choice clearly complicates the derivation of the conditional score, raising the question whether a simple independence between these two random vectors could have been used instead. In this appendix we explore this option and expose its limitation.

We start by defining $\mathbf{\tilde{x}}_{L+1} = \mathbf{x}$, and for every ${i = L, L-1, \dots, 1}$: ${\mathbf{\tilde{x}}_{i} = \mathbf{\tilde{x}}_{i+1} + \text{\boldmath$\eta$}_i}$, where ${\text{\boldmath$\eta$}_i \sim \mathcal{N}\left(0, \left( \sigma_i^2 - \sigma_{i+1}^2 \right) \mathbf{I}\right)}$ is independent of $\mathbf{z}$.
This results in $\mathbf{\tilde{x}}_{i} = \mathbf{x} + \mathbf{n}_i$, where ${\mathbf{n}_i = \sum_{k=i}^{L} \text{\boldmath$\eta$}_k \sim \mathcal{N}\left(0, \sigma_i^2 \mathbf{I}\right)}$.
As before, we aim to derive the conditional score function $p\left(\mathbf{\tilde{x}}_i|\mathbf{y}\right)$ and thus we look at the vector
\begin{equation}
\label{eqn:apdx_alt}
    \mathbf{y} - \mathbf{H\tilde{x}}_i = \mathbf{Hx} + \mathbf{z} - \mathbf{Hx} - \mathbf{Hn}_i = \mathbf{z} - \mathbf{Hn}_{i}.
\end{equation}
This is a Gaussian vector with zero mean and a covariance matrix $\sigma_0^2 \mathbf{I} + \sigma_i^2 \mathbf{H H}^T$, due to the independence between $\mathbf{z}$ and $\mathbf{n}$.
In order to make use of \autoref{eqn:apdx_alt}, we would like to express $p\left(\mathbf{\tilde{x}}_i | \mathbf{y}\right)$ as $p\left(\mathbf{H\tilde{x}}_i - \mathbf{y} | \mathbf{y}\right)$. However, this transition is not possible because the multiplication by $\mathbf{H}$ is not an invertible operation, which means that it changes the statistics of the tested vector.
Instead, $p\left(\mathbf{\tilde{x}}_i | \mathbf{y}\right)$ may be expressed using the Bayes rule as
\[
    p\left(\mathbf{\tilde{x}}_i | \mathbf{y}\right) = \frac{1}{p\left(\mathbf{y}\right)} p\left(\mathbf{\tilde{x}}_i\right) p\left(\mathbf{y} | \mathbf{\tilde{x}}_i\right) = \frac{1}{p\left(\mathbf{y}\right)} p\left(\mathbf{\tilde{x}}_i\right) p\left(\mathbf{y} - \mathbf{H\tilde{x}}_i | \mathbf{\tilde{x}}_i\right).
\]
The first term $1 / p\left(\mathbf{y}\right)$ becomes zero after differentiating by $\mathbf{\tilde{x}}_i$, and the second term's gradient log can be approximated using a neural network, as done before.
The third term describes a Gaussian vector, and can be written as $p\left( \mathbf{z} - \mathbf{Hn}_{i} | \mathbf{\tilde{x}}_i\right)$ due to \autoref{eqn:apdx_alt}. The Gaussian vector $\mathbf{z} - \mathbf{Hn}_{i}$ is conditioned on $\mathbf{\tilde{x}}_i = \mathbf{x}_i + \mathbf{n}_i$, which encapsulates information about $\mathbf{n}_i$, without a clear way of knowing $\mathbf{n}_i$ itself. Thus, without an explicit term for $p\left(\mathbf{n}_i | \mathbf{\tilde{x}}_i\right)$, we are unable to derive an analytical term for the gradient log of the likelihood.

Therefore, the path we took to define the noise additions aims for the difference $\mathbf{y} - \mathbf{H\tilde{x}}_i$ to be independent of $\mathbf{\tilde{x}}_i$. In order to achieve that, we use the SVD decomposition of $\mathbf{H}$ and define the noise addition sequence as in~\autoref{sec:problem}, both steps seem unavoidable.
\section{Implementation Details}
We run SNIPS with the hyperparameters detailed in \autoref{tab:params}, where $\left\{\sigma_{i}\right\}_{i=1}^{L}$ is a decreasing geometric sequence.
These hyperparameters conform to those used in NCSNv2~\cite{song2020improved}, the neural network model that we used.
The parameters $\mathbf{H}$, $\sigma_0$ and $\mathbf{y}$ are defined by the inverse problem at hand.
Recall that this algorithm applies $\tau L$ overall iterations to complete, in each a denoiser is being activated.
The sampling algorithm was run on a single Nvidia RTX3080 GPU with 10GB memory, and took around $2$ minutes for producing $8$ samples from the $64 \times 64$ CelebA dataset, and around $6$ minutes for producing $6$ samples from the $128 \times 128$ LSUN dataset. The exact times vary slightly for the various inverse problems.

The code used in this paper is available at \url{https://github.com/bahjat-kawar/snips_torch}.

\begin{table}
    \centering
    \caption{Hyperparameters for our experiments. ${\frac{\sigma_{i+1}}{\sigma_{i}}}$ is the geometric common ratio for $\left\{\sigma_{i}\right\}_{i=1}^{L}$.
    }
    \label{tab:params}
    \begin{tabular}{c c c c c c c}
        \toprule
        \textbf{Dataset} & ${c}$ & ${\tau}$ & ${L}$ & ${\sigma_{1}}$ & ${\sigma_{L}}$ & ${\frac{\sigma_{i+1}}{\sigma_{i}}}$ \\
        \midrule
        \textbf{CelebA}
        &$3.3e-2$ & ${5}$ & $500$ & ${90}$ & $0.01$ & $0.982$\\
        \textbf{LSUN}
        &$1.8e-2$ & ${3}$ & $1086$ & ${190}$ & $0.01$ & $0.991$\\
        \bottomrule
    \end{tabular}
\end{table}

\section{Comparison to RED}
\label{sec:red}
\begin{table}
    \centering
    \caption{Comparison between SNIPS and RED on 8 CelebA images. SNIPS Mean is the average of 8 SNIPS outputs per image. The best number in each row is in \textbf{bold}.}
    \label{tab:red}
    \begin{tabular}{c c c c c}
        \toprule
        \textbf{Problem} & \textbf{Metric} & \textbf{SNIPS} & \textbf{SNIPS Mean} & \textbf{RED} \\
        \midrule
        \multirow{2}{*}{ {Uniform deblurring with $\sigma_0 = 0.006$}}
        & {PSNR $\uparrow$} & $32.41$ & $\mathbf{35.42}$ & $29.03$ \\
        & {LPIPS $\downarrow$} & $\mathbf{0.005}$ & $\mathbf{0.005}$ & $0.043$ \\
        \multirow{2}{*}{ {Uniform deblurring with $\sigma_0 = 0.1$}}
        & {PSNR $\uparrow$} & $25.03$ & $\mathbf{27.28}$ & $20.10$ \\
        & {LPIPS $\downarrow$} & $\mathbf{0.032}$ & $0.045$ & $0.077$ \\
        \bottomrule
    \end{tabular}
\end{table}

RED~\cite{romano2017red} is a well-known method that leverages a denoiser for the MAP solution of inverse problems, and as such it is a relevant method to compare with. We compare RED to SNIPS on the image deblurring problem (with a uniform $5 \times 5$ kernel and additive noise with $\sigma_0$), while using the same denoiser model (NCSNv2) for both. We run the SD (Steepest Descent) version of RED on the luminance channel of the image in the YCbCr color space, as in the original paper, with its hyperparameters chosen for best PSNR performance. Namely, $\lambda = 0.12, N = 100$ for $\sigma_0 = 0.006$, and $\lambda = 1000, N = 100$ for $\sigma_0 = 0.1$. In addition to PSNR, we also calculate LPIPS~\cite{lpips}, a perceptual quality metric, in order to verify the claim that SNIPS has superior visual quality.

As can be seen in \autoref{tab:red}, both SNIPS and its mean outperform RED in PSNR as well as LPIPS. When the noise is significant ($\sigma_0 = 0.1$), it becomes clear that SNIPS has superior visual quality at the expense of PSNR performance, in comparison to the average of samples. A visual comparison is shown in \autoref{fig:red}.

\begin{figure}
    \centering
    \def\arraystretch{0.6}
    \setlength\tabcolsep{1pt}
    \begin{tabular}{ccccc}
        \includegraphics[width=2cm,height=2cm]{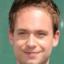}
        & \includegraphics[width=2cm,height=2cm]{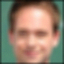}
        & \includegraphics[width=2cm,height=2cm]{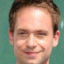}
        & \includegraphics[width=2cm,height=2cm]{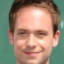}
        & \includegraphics[width=2cm,height=2cm]{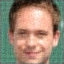}
        \\
        \includegraphics[width=2cm,height=2cm]{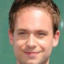}
        & \includegraphics[width=2cm,height=2cm]{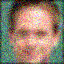}
        & \includegraphics[width=2cm,height=2cm]{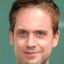}
        & \includegraphics[width=2cm,height=2cm]{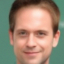}
        & \includegraphics[width=2cm,height=2cm]{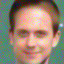}
        \\
        \scriptsize{Original} & \scriptsize{Blurred} & \scriptsize{SNIPS} & \scriptsize{SNIPS Mean} & \scriptsize{RED}
    \end{tabular}
    \caption{Deblurring results on a CelebA image (uniform $5 \times 5$ blur). Top: additive noise with $\sigma_0=0.006$, bottom: additive noise with $\sigma_0=0.1$.}
    \label{fig:red}
\end{figure}

\section{Additional Results}
We provide below more results of SNIPS for image deblurring, super-resolution and compressive sampling. We recommend to view these figures zoomed-in in order to see the details in the produced samples (or lack thereof in their average).

\begin{figure}[h]
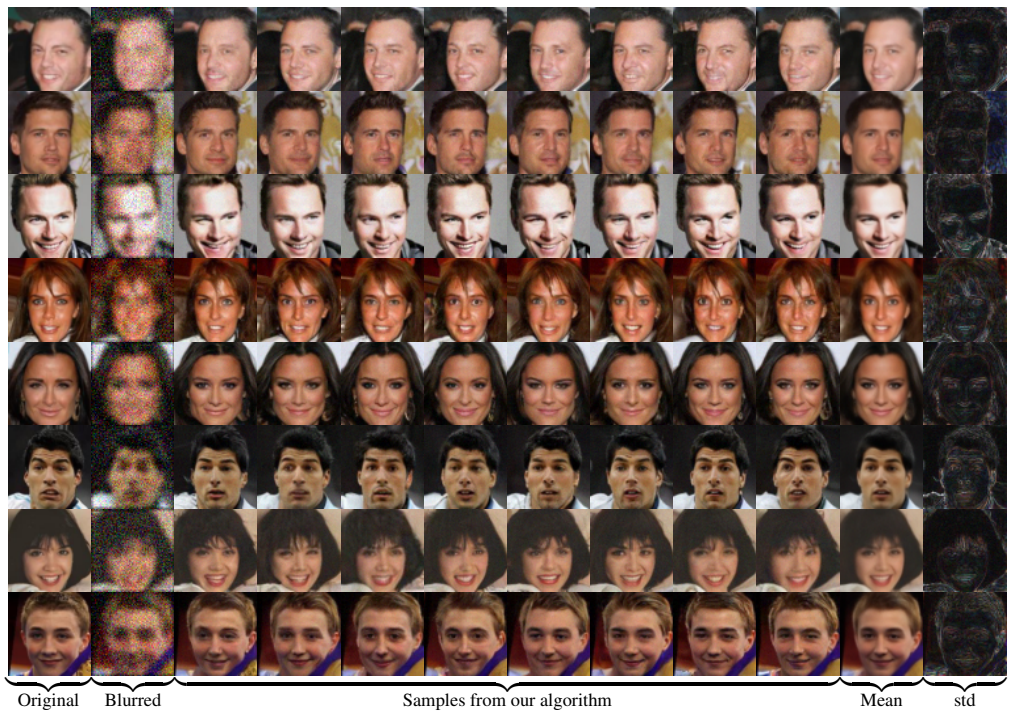

    \centering
    \def\arraystretch{0.1}
    \setlength\tabcolsep{0pt}
    \begin{tabular}{cccccccccccc}
        \forloop{row}{0}{\value{row} < 8}{
            \hspace{-0.05cm}\includegraphics[width=1.1cm,height=1.1cm]{./images/unideblur/sample_\arabic{row}_column_0.png} \hspace{-0.17cm}
            \forloop{col}{1}{\value{col} < 12}{
                & \includegraphics[width=1.1cm,height=1.1cm]{./images/unideblur/sample_\arabic{row}_column_\arabic{col}.png} \hspace{-0.17cm}
            } \\
        }
        & \multicolumn{11}{c}{\vspace{0.5mm}}\\
        \upbracefill & \upbracefill &
        \multicolumn{8}{c}{
            \upbracefill
        } &
        \upbracefill & \upbracefill
        \\
        & \multicolumn{11}{c}{\vspace{0.5mm}}\\
        \scriptsize{Original} & \scriptsize{Blurred} &
        \multicolumn{8}{c}{
            \scriptsize{Samples from our algorithm}
        } &
        \scriptsize{Mean} & \scriptsize{std}
    \end{tabular}
    \caption{Deblurring results on CelebA images (uniform $5 \times 5$ blur and an additive noise with $\sigma_0 = 0.1$).}
\end{figure}

\newcounter{main}
\begin{figure}
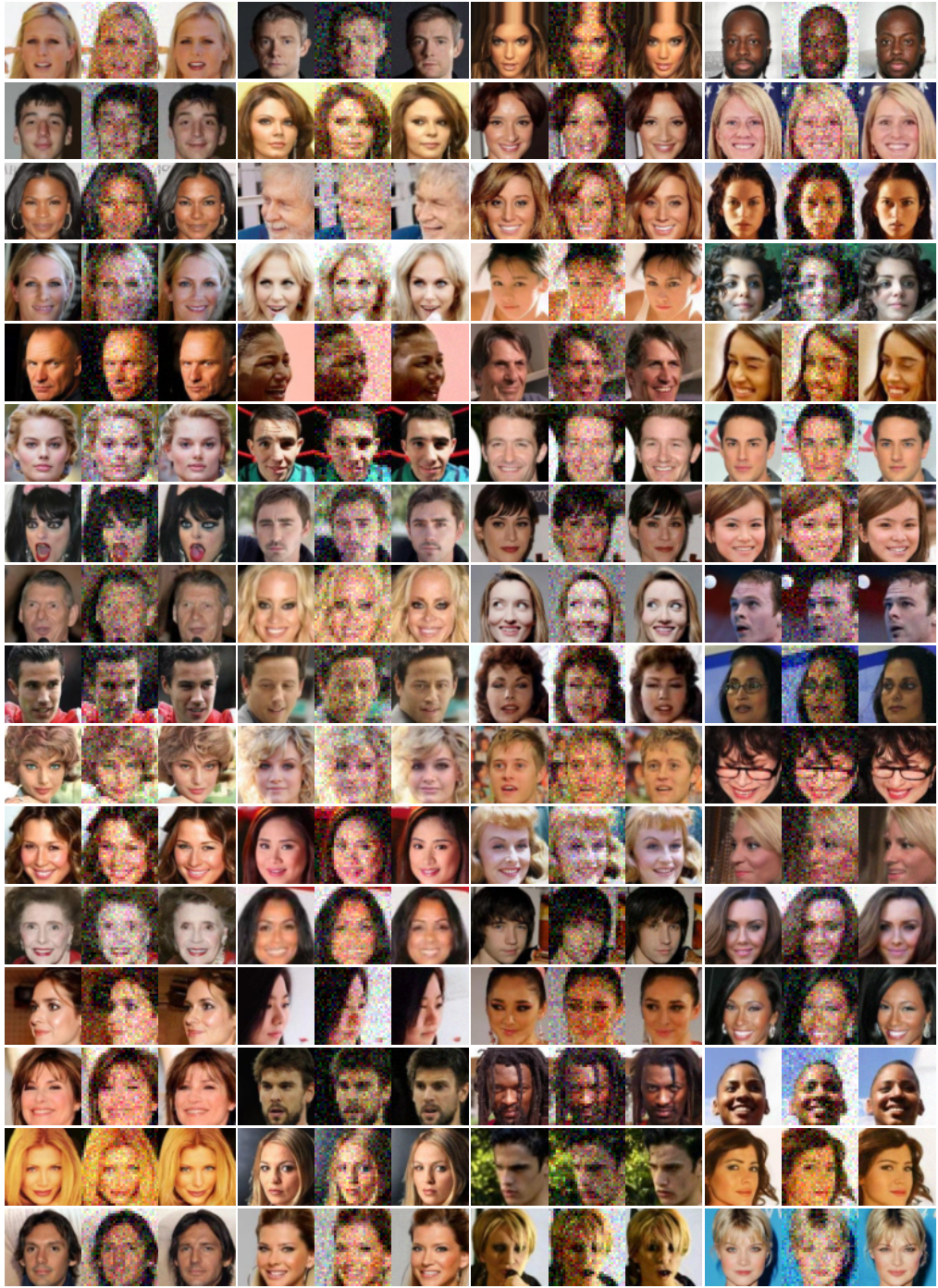

    \centering
    \def\arraystretch{0.5}
    \setlength\tabcolsep{0pt}
    \begin{tabular}{cccc}
    \begin{tabular}{ccc}
        \forloop{row}{0}{\value{row} < 16}{
            \hspace{-0.05cm}\includegraphics[width=1.1cm,height=1.1cm]{./images/extra_sr2/sample_\arabic{row}_column_0.png} \hspace{-0.17cm}
            \forloop{col}{1}{\value{col} < 3}{
                & \includegraphics[width=1.1cm,height=1.1cm]{./images/extra_sr2/sample_\arabic{row}_column_\arabic{col}.png} \hspace{-0.17cm}
            } \\
        }
    \end{tabular}
    &
    \begin{tabular}{ccc}
        \forloop{row}{16}{\value{row} < 32}{
            \hspace{-0.05cm}\includegraphics[width=1.1cm,height=1.1cm]{./images/extra_sr2/sample_\arabic{row}_column_0.png} \hspace{-0.17cm}
            \forloop{col}{1}{\value{col} < 3}{
                & \includegraphics[width=1.1cm,height=1.1cm]{./images/extra_sr2/sample_\arabic{row}_column_\arabic{col}.png} \hspace{-0.17cm}
            } \\
        }
    \end{tabular}
    &
    \begin{tabular}{ccc}
        \forloop{row}{32}{\value{row} < 48}{
            \hspace{-0.05cm}\includegraphics[width=1.1cm,height=1.1cm]{./images/extra_sr2/sample_\arabic{row}_column_0.png} \hspace{-0.17cm}
            \forloop{col}{1}{\value{col} < 3}{
                & \includegraphics[width=1.1cm,height=1.1cm]{./images/extra_sr2/sample_\arabic{row}_column_\arabic{col}.png} \hspace{-0.17cm}
            } \\
        }
    \end{tabular}
    &
    \begin{tabular}{ccc}
        \forloop{row}{48}{\value{row} < 64}{
            \hspace{-0.05cm}\includegraphics[width=1.1cm,height=1.1cm]{./images/extra_sr2/sample_\arabic{row}_column_0.png} \hspace{-0.17cm}
            \forloop{col}{1}{\value{col} < 3}{
                & \includegraphics[width=1.1cm,height=1.1cm]{./images/extra_sr2/sample_\arabic{row}_column_\arabic{col}.png} \hspace{-0.17cm}
            } \\
        }
    \end{tabular}
    \end{tabular}
    \caption{Extended uncurated super resolution results on CelebA images (downscaling $2:1$ by plain averaging and adding noise with $\sigma_0 = 0.1$). Every image set contains: original, low-res, SNIPS restoration, in that order.}
\end{figure}

\begin{figure}
    \centering
    \def\arraystretch{0.1}
    \setlength\tabcolsep{0pt}
    \begin{tabular}{cccccccccccc}
        \forloop{row}{0}{\value{row} < 8}{
            \hspace{-0.05cm}\includegraphics[width=1.1cm,height=1.1cm]{./images/sr2/sample_\arabic{row}_column_0.png} \hspace{-0.17cm}
            \forloop{col}{1}{\value{col} < 12}{
                & \includegraphics[width=1.1cm,height=1.1cm]{./images/sr2/sample_\arabic{row}_column_\arabic{col}.png} \hspace{-0.17cm}
            } \\
        }
        & \multicolumn{11}{c}{\vspace{0.5mm}}\\
        \upbracefill & \upbracefill &
        \multicolumn{8}{c}{
            \upbracefill
        } &
        \upbracefill & \upbracefill
        \\
        & \multicolumn{11}{c}{\vspace{0.5mm}}\\
        \scriptsize{Original} & \scriptsize{Low-res} &
        \multicolumn{8}{c}{
            \scriptsize{Samples from our algorithm}
        } &
        \scriptsize{Mean} & \scriptsize{std}
    \end{tabular}
    \caption{Super resolution results on CelebA images (downscaling $2:1$ by plain averaging and adding noise with $\sigma_0 = 0.1$).}
\end{figure}

\begin{figure}
    \centering
    \def\arraystretch{0.1}
    \setlength\tabcolsep{0pt}
    \begin{tabular}{cccccccccccc}
        \forloop{row}{0}{\value{row} < 8}{
            \hspace{-0.05cm}\includegraphics[width=1.1cm,height=1.1cm]{./images/sr4/sample_\arabic{row}_column_0.png} \hspace{-0.17cm}
            \forloop{col}{1}{\value{col} < 12}{
                & \includegraphics[width=1.1cm,height=1.1cm]{./images/sr4/sample_\arabic{row}_column_\arabic{col}.png} \hspace{-0.17cm}
            } \\
        }
        & \multicolumn{11}{c}{\vspace{0.5mm}}\\
        \upbracefill & \upbracefill &
        \multicolumn{8}{c}{
            \upbracefill
        } &
        \upbracefill & \upbracefill
        \\
        & \multicolumn{11}{c}{\vspace{0.5mm}}\\
        \scriptsize{Original} & \scriptsize{Low-res} &
        \multicolumn{8}{c}{
            \scriptsize{Samples from our algorithm}
        } &
        \scriptsize{Mean} & \scriptsize{std}
    \end{tabular}
    \caption{Super resolution results on CelebA images (downscaling $4:1$ by plain averaging and adding noise with $\sigma_0 = 0.1$).}
\end{figure}

\begin{figure}
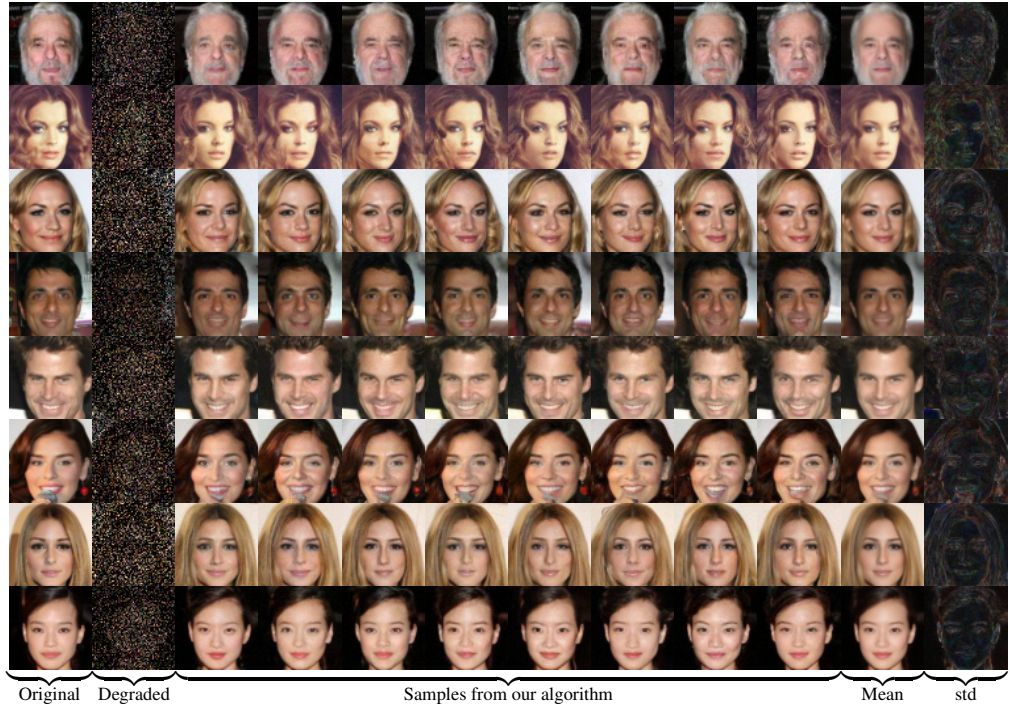

    \centering
    \def\arraystretch{0.1}
    \setlength\tabcolsep{0pt}
    \begin{tabular}{cccccccccccc}
        \forloop{row}{0}{\value{row} < 8}{
            \hspace{-0.05cm}\includegraphics[width=1.1cm,height=1.1cm]{./images/cs4/sample_\arabic{row}_column_0.png} \hspace{-0.17cm}
            \forloop{col}{1}{\value{col} < 12}{
                & \includegraphics[width=1.1cm,height=1.1cm]{./images/cs4/sample_\arabic{row}_column_\arabic{col}.png} \hspace{-0.17cm}
            } \\
        }
        & \multicolumn{11}{c}{\vspace{0.5mm}}\\
        \upbracefill & \upbracefill &
        \multicolumn{8}{c}{
            \upbracefill
        } &
        \upbracefill & \upbracefill
        \\
        & \multicolumn{11}{c}{\vspace{0.5mm}}\\
        \scriptsize{Original} & \scriptsize{Degraded} &
        \multicolumn{8}{c}{
            \scriptsize{Samples from our algorithm}
        } &
        \scriptsize{Mean} & \scriptsize{std}
    \end{tabular}
    \caption{Compressive sensing results on CelebA images (compression by $25\%$ and adding noise with $\sigma_0 = 0.1$).}
\end{figure}

\begin{figure}
    \centering
    \def\arraystretch{0.1}
    \setlength\tabcolsep{0pt}
    \begin{tabular}{cccccccc}
        \forloop{row}{0}{\value{row} < 6}{
            \hspace{-0.05cm}\includegraphics[width=1.7cm,height=1.7cm]{./images/bedroom_sr2/sample_\arabic{row}_column_0.png} \hspace{-0.17cm}
            \forloop{col}{1}{\value{col} < 8}{
                & \includegraphics[width=1.7cm,height=1.7cm]{./images/bedroom_sr2/sample_\arabic{row}_column_\arabic{col}.png} \hspace{-0.17cm}
            } \\
        }
        & \multicolumn{7}{c}{\vspace{0.5mm}}\\
        \upbracefill & \upbracefill &
        \multicolumn{4}{c}{
            \upbracefill
        } &
        \upbracefill & \upbracefill
        \\
        & \multicolumn{7}{c}{\vspace{0.5mm}}\\
        \scriptsize{Original} & \scriptsize{Low-res} &
        \multicolumn{4}{c}{
            \scriptsize{Samples from our algorithm}
        } &
        \scriptsize{Mean} & \scriptsize{std}
    \end{tabular}
    \caption{Super resolution results on LSUN bedroom images (downscaling $2:1$ by plain averaging and adding noise with $\sigma_0 = 0.04$).}
\end{figure}

\begin{figure}
    \centering
    \def\arraystretch{0.1}
    \setlength\tabcolsep{0pt}
    \begin{tabular}{cccccccc}
        \forloop{row}{0}{\value{row} < 6}{
            \hspace{-0.05cm}\includegraphics[width=1.7cm,height=1.7cm]{./images/bedroom_sr4/sample_\arabic{row}_column_0.png} \hspace{-0.17cm}
            \forloop{col}{1}{\value{col} < 8}{
                & \includegraphics[width=1.7cm,height=1.7cm]{./images/bedroom_sr4/sample_\arabic{row}_column_\arabic{col}.png} \hspace{-0.17cm}
            } \\
        }
        & \multicolumn{7}{c}{\vspace{0.5mm}}\\
        \upbracefill & \upbracefill &
        \multicolumn{4}{c}{
            \upbracefill
        } &
        \upbracefill & \upbracefill
        \\
        & \multicolumn{7}{c}{\vspace{0.5mm}}\\
        \scriptsize{Original} & \scriptsize{Low-res} &
        \multicolumn{4}{c}{
            \scriptsize{Samples from our algorithm}
        } &
        \scriptsize{Mean} & \scriptsize{std}
    \end{tabular}
    \caption{Super resolution results on LSUN bedroom images (downscaling $4:1$ by plain averaging and adding noise with $\sigma_0 = 0.04$).}
\end{figure}

\begin{figure}
    \centering
    \def\arraystretch{0.1}
    \setlength\tabcolsep{0pt}
    \begin{tabular}{cccccccc}
        \forloop{row}{0}{\value{row} < 6}{
            \hspace{-0.05cm}\includegraphics[width=1.7cm,height=1.7cm]{./images/bedroom_cs4/sample_\arabic{row}_column_0.png} \hspace{-0.17cm}
            \forloop{col}{1}{\value{col} < 8}{
                & \includegraphics[width=1.7cm,height=1.7cm]{./images/bedroom_cs4/sample_\arabic{row}_column_\arabic{col}.png} \hspace{-0.17cm}
            } \\
        }
        & \multicolumn{7}{c}{\vspace{0.5mm}}\\
        \upbracefill & \upbracefill &
        \multicolumn{4}{c}{
            \upbracefill
        } &
        \upbracefill & \upbracefill
        \\
        & \multicolumn{7}{c}{\vspace{0.5mm}}\\
        \scriptsize{Original} & \scriptsize{Low-res} &
        \multicolumn{4}{c}{
            \scriptsize{Samples from our algorithm}
        } &
        \scriptsize{Mean} & \scriptsize{std}
    \end{tabular}
    \caption{Compressive sensing results on LSUN bedroom images (compression by $25\%$ and adding noise with $\sigma_0 = 0.04$).}
\end{figure}

\begin{figure}
    \centering
    \def\arraystretch{0.1}
    \setlength\tabcolsep{0pt}
    \begin{tabular}{cccccccc}
        \forloop{row}{0}{\value{row} < 6}{
            \hspace{-0.05cm}\includegraphics[width=1.7cm,height=1.7cm]{./images/tower_sr2/sample_\arabic{row}_column_0.png} \hspace{-0.17cm}
            \forloop{col}{1}{\value{col} < 8}{
                & \includegraphics[width=1.7cm,height=1.7cm]{./images/tower_sr2/sample_\arabic{row}_column_\arabic{col}.png} \hspace{-0.17cm}
            } \\
        }
        & \multicolumn{7}{c}{\vspace{0.5mm}}\\
        \upbracefill & \upbracefill &
        \multicolumn{4}{c}{
            \upbracefill
        } &
        \upbracefill & \upbracefill
        \\
        & \multicolumn{7}{c}{\vspace{0.5mm}}\\
        \scriptsize{Original} & \scriptsize{Low-res} &
        \multicolumn{4}{c}{
            \scriptsize{Samples from our algorithm}
        } &
        \scriptsize{Mean} & \scriptsize{std}
    \end{tabular}
    \caption{Super resolution results on LSUN tower images (downscaling $2:1$ by plain averaging and adding noise with $\sigma_0 = 0.04$).}
\end{figure}

\begin{figure}
    \centering
    \def\arraystretch{0.1}
    \setlength\tabcolsep{0pt}
    \begin{tabular}{cccccccc}
        \forloop{row}{0}{\value{row} < 6}{
            \hspace{-0.05cm}\includegraphics[width=1.7cm,height=1.7cm]{./images/tower_sr4/sample_\arabic{row}_column_0.png} \hspace{-0.17cm}
            \forloop{col}{1}{\value{col} < 8}{
                & \includegraphics[width=1.7cm,height=1.7cm]{./images/tower_sr4/sample_\arabic{row}_column_\arabic{col}.png} \hspace{-0.17cm}
            } \\
        }
        & \multicolumn{7}{c}{\vspace{0.5mm}}\\
        \upbracefill & \upbracefill &
        \multicolumn{4}{c}{
            \upbracefill
        } &
        \upbracefill & \upbracefill
        \\
        & \multicolumn{7}{c}{\vspace{0.5mm}}\\
        \scriptsize{Original} & \scriptsize{Low-res} &
        \multicolumn{4}{c}{
            \scriptsize{Samples from our algorithm}
        } &
        \scriptsize{Mean} & \scriptsize{std}
    \end{tabular}
    \caption{Super resolution results on LSUN tower images (downscaling $4:1$ by plain averaging and adding noise with $\sigma_0 = 0.04$).}
\end{figure}

\begin{figure}
    \centering
    \def\arraystretch{0.1}
    \setlength\tabcolsep{0pt}
    \begin{tabular}{cccccccc}
        \forloop{row}{0}{\value{row} < 6}{
            \hspace{-0.05cm}\includegraphics[width=1.7cm,height=1.7cm]{./images/tower_cs4/sample_\arabic{row}_column_0.png} \hspace{-0.17cm}
            \forloop{col}{1}{\value{col} < 8}{
                & \includegraphics[width=1.7cm,height=1.7cm]{./images/tower_cs4/sample_\arabic{row}_column_\arabic{col}.png} \hspace{-0.17cm}
            } \\
        }
        & \multicolumn{7}{c}{\vspace{0.5mm}}\\
        \upbracefill & \upbracefill &
        \multicolumn{4}{c}{
            \upbracefill
        } &
        \upbracefill & \upbracefill
        \\
        & \multicolumn{7}{c}{\vspace{0.5mm}}\\
        \scriptsize{Original} & \scriptsize{Low-res} &
        \multicolumn{4}{c}{
            \scriptsize{Samples from our algorithm}
        } &
        \scriptsize{Mean} & \scriptsize{std}
    \end{tabular}
    \caption{Compressive sensing results on LSUN tower images (compression by $25\%$ and adding noise with $\sigma_0 = 0.04$).}
\end{figure}

\end{document}